\documentclass[acmsmall]{acmart}

\usepackage{graphicx}
\usepackage{balance}  

\usepackage{graphics}
\usepackage{mdwlist}
\usepackage{algorithm}
\usepackage[noend]{algpseudocode}

\usepackage{booktabs}
\usepackage[center]{subfigure}

\usepackage{amsthm}
\usepackage[]{caption}
\usepackage{multirow}
\usepackage{mathrsfs}
\usepackage{amsbsy}
\usepackage[mathscr]{eucal} 
\usepackage{xspace}
\usepackage{paralist} 
\usepackage{mathtools}
\usepackage{verbatim}
\usepackage{enumitem}
\usepackage{pifont}

\newtheorem{problem}{Problem}
\newtheorem{defn}{Definition}
\newtheorem{lem}{Lemma}
\newtheorem{thm}{Theorem}

\newenvironment{sproof}{%
	\renewcommand{\proofname}{Sketch of Proof}\proof}{\endproof}

\newcommand{\arxivD}{{Arxiv}\xspace}

\newcommand{\facebookD}{{Facebook}\xspace}
\newcommand{\youtubeD}{{Youtube}\xspace}
\newcommand{\randomD}{{Random}\xspace}

\newcommand{\flickrD}{{Flickr}\xspace}
\newcommand{\berkstanD}{{BerkStan}\xspace}
\newcommand{\friendsterD}{{Friendster}\xspace}

\newcommand{\livejournalD}{{LiveJournal}\xspace}
\newcommand{\googleD}{{Google}\xspace}

\newcommand{\smallsection}[1]{{\vspace{0.1in} \noindent {\bf{\underline{\smash{#1}}}}}}
\newcommand{\smallsectionnoline}[1]{{\vspace{0.1in} \noindent {\bf{#1}}}}

\newcommand{\figsummary}[1]{\textbf{#1}}

\newcommand{\SG}{\mathcal{G}}
\newcommand{\SV}{\mathcal{V}}
\newcommand{\SE}{\mathcal{E}}

\newcommand{\BE}{\mathbb{E}}
\newcommand{\BP}{\mathbb{P}}

\newcommand{\SGT}{\SG^{(t)}}
\newcommand{\SVT}{\SV^{(t)}}
\newcommand{\SET}{\SE^{(t)}}

\newcommand{\zt}{z^{(t)}}

\newcommand{\cmark}{\ding{51}}%
\newcommand{\bit}{\begin{itemize}[leftmargin=15pt]}
\newcommand{\eit}{\end{itemize}}
\newcommand{\ben}{\begin{enumerate}[leftmargin=15pt]}
\newcommand{\een}{\end{enumerate}}

\newcommand{\mapreduce}{\textsc{MapReduce}\xspace}

\newcommand{\wrs}{\textsc{WRS}\xspace}
\newcommand{\trifly}{\textsc{Tri-Fly}\xspace}
\newcommand{\cocos}{\textsc{CoCoS}\xspace}
\newcommand{\cocosopt}{\textsc{CoCoS}\textsubscript{OPT}\xspace}
\newcommand{\cocossimple}{\textsc{CoCoS}\textsubscript{SIMPLE}\xspace}
\newcommand{\thinkd}{\textsc{ThinkD}\xspace}

\DeclareMathOperator*{\argmin}{arg\,min}

\newcommand{\ints}{\{1,2,...\}}

\newcommand{\uv}{\{u,v\}}
\newcommand{\vw}{\{v,w\}}
\newcommand{\wu}{\{w,u\}}

\newcommand{\pair}{\uv}

\newcommand{\et}{e^{(t)}}
\newcommand{\eone}{e^{(1)}}
\newcommand{\etwo}{e^{(2)}}
\newcommand{\gstream}{(\eone, \etwo, ...)}

\newcommand{\triestimp}{\textsc{Triest\textsubscript{impr}}\xspace}

\newcommand{\mascot}{\textsc{Mascot}\xspace}
\newcommand{\rept}{\textsc{REPT}\xspace}

\newcommand{\rmse}{RMSE\xspace}

\newcommand{\globalerrorL}{Global Error\xspace}
\newcommand{\globalvarL}{Global Variance\xspace}
\newcommand{\localerrorL}{Local Error\xspace}

\newcommand{\rankcorrelationL}{Rank Correlation\xspace}

\newcommand{\budget}{b}
\newcommand{\triple}{\{u,v,w\}}
\newcommand{\tuv}{t_{uv}}
\newcommand{\tvw}{t_{vw}}
\newcommand{\twu}{t_{wu}}

\newcommand{\pt}{p^{(t)}}
\newcommand{\STT}{\mathcal{T}^{(t)}}
\newcommand{\STTu}{\STT[u]}
\newcommand{\cbar}{\bar{c}}
\newcommand{\cbart}{\bar{c}^{(t)}}

\newcommand{\cu}{c[u]}
\newcommand{\ctu}{c^{(t)}[u]}
\newcommand{\globalnum}{|\STT|}
\newcommand{\localnum}{|\STTu|}

\newcommand{\es}{e^{(s)}}

\newcommand{\tripletwo}{\{u,v,x\}}

\newcommand{\qt}{q^{(t)}}
\newcommand{\istar}{i^{*}}
\newcommand{\li}{l_{i}}
\newcommand{\lj}{l_{j}}
\newcommand{\lucky}{\textsc{Lucky}\xspace}
\newcommand{\unlucky}{\textsc{Unlucky}\xspace}
\newcommand{\assign}{\textsc{Assigned}\xspace}
\newcommand{\noassign}{\textsc{Unassigned}\xspace}
\newcommand{\workernum}{k}
\newcommand{\workerset}{\{1,...,\workernum\}}
\newcommand{\SGI}{\mathcal{G}_{i}}
\newcommand{\SVI}{\mathcal{V}_{i}}
\newcommand{\SEI}{\mathcal{E}_{i}}

\newcommand{\SESI}{\mathcal{E}^{(s)}_{i}}
\newcommand{\SNI}{\mathcal{N}_{i}}

\newcommand{\lti}{l_{i}^{(t)}}
\newcommand{\ltfu}{l_{f(u)}^{(t)}}
\newcommand{\ltip}{l_{i}^{(t+1)}}
\newcommand{\lsi}{l_{i}^{(s)}}
\newcommand{\pti}{p_{i}^{(t)}}
\newcommand{\qti}{q_{i}^{(t)}}
\newcommand{\piuvw}{p_{i}[uvw]}
\newcommand{\diuvw}{d_{i}[uvw]}
\newcommand{\fuvw}{f(uvw)}
\newcommand{\pfuvwuvw}{p_{\fuvw}[uvw]}
\newcommand{\dfuvwuvw}{d_{\fuvw}[uvw]}

\newcommand{\listar}{l_{\istar}}
\newcommand{\SEfuvw}{\mathcal{E}_{\fuvw}}
\newcommand{\sbar}{S}
\newcommand{\lbar}{L}
\newcommand{\zti}{z_{i}^{(t)}}
\newcommand{\STTI}{\STT_{i}}
\newcommand{\globalnumi}{|\STTI|}
\newcommand{\cbarit}{\cbart_{i}}
\newcommand{\ctiu}{c^{(t)}_{i}[u]}
\newcommand{\ctv}{c^{(t)}[v]}
\newcommand{\ctw}{c^{(t)}[w]}

\newcommand{\gsh}{\textsc{GSH$_{T}$}\xspace}
\newcommand{\ns}{\textsc{NS}\xspace}
\AtBeginDocument{%
  \providecommand\BibTeX{{%
    \normalfont B\kern-0.5em{\scshape i\kern-0.25em b}\kern-0.8em\TeX}}}

\acmJournal{TKDD}

\begin{document}

\title{CoCoS: Fast and Accurate Distributed Triangle Counting in Graph Streams}

\author{Kijung Shin}
\affiliation{%
	\institution{KAIST}
	\city{Daejeon}
	\country{South Korea}
	\postcode{34141}
}
\email{kjiungs@kaist.ac.kr}
\orcid{0000-0002-2872-1526}

\author{Euiwoong Lee}
\affiliation{%
  \institution{University of Michigan}
  \city{Ann Arbor}
  \state{MI}
  \country{United States}
  \postcode{48109}
}
\email{euiwoong@umich.edu}

\author{Jinoh Oh}
\affiliation{%
	\institution{Carnegie Mellon University}
	\city{Pittsburgh}
	\state{PA}
	\country{United States}
	\postcode{15213}
}
\email{jinoho@cs.cmu.edu}

\author{Mohammad Hammoud}
\affiliation{%
	\institution{Carnegie Mellon University in Qatar}
	\city{Doha}
	\country{Qatar}
	\postcode{24866}
}
\email{mhhamoud@cmu.edu}

\author{Christos Faloutsos}
\affiliation{%
	\institution{Carnegie Mellon University}
	\city{Pittsburgh}
	\state{PA}
	\country{United States}
	\postcode{15213}
}
\email{christos@cs.cmu.edu}

\renewcommand{\shortauthors}{Shin et al}

\begin{abstract}
	\textit{Given a graph stream, how can we estimate the number of triangles
in it using multiple machines with limited storage? 
Specifically, how should edges be processed and sampled across the machines
for rapid and accurate estimation?}

The count of triangles (i.e., cliques of size three) has proven useful in numerous applications, including anomaly detection, community detection, and link recommendation.
For triangle counting in large and dynamic graphs,
recent work has focused largely on streaming algorithms and distributed algorithms but little on their combinations for ``the best of both worlds''.

In this work, we propose \cocos, a fast and accurate distributed streaming algorithm for estimating the counts of global triangles (i.e., all triangles) and local triangles incident to each node. Making one pass over the input stream, \cocos carefully processes and stores the edges across multiple machines so that the redundant use of computational and storage resources is minimized.
Compared to baselines, \cocos is {\it (a) Accurate:} giving up to {\bf $\mathbf{39\times}$ smaller estimation error}, {\it (b) Fast}: up to {\bf $\mathbf{10.4\times}$ faster}, scaling linearly with the size of the input stream, and {\it (c) Theoretically sound}: yielding unbiased estimates. 
\end{abstract}



\maketitle

\section{Introduction}
\label{sec:trifly:intro}
\begin{figure}[t]
	\centering
	\hspace{-4mm}
	\subfigure[Fast and accurate]{
		\includegraphics[width= 0.235\linewidth]{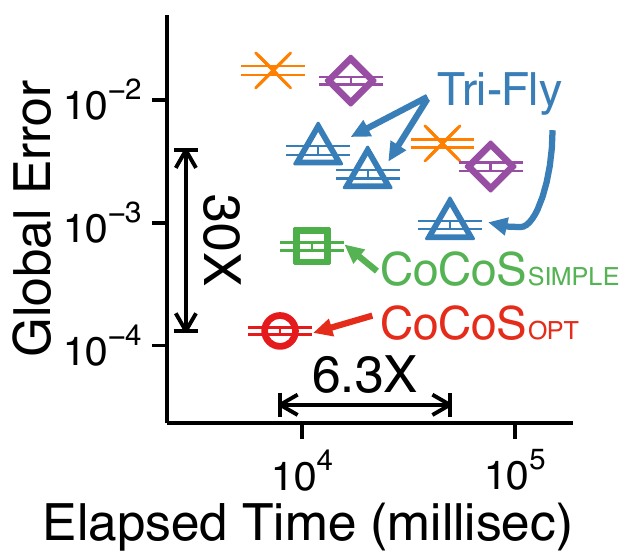}
		\label{fig:trifly:crown:tradeoff}
	}
	\subfigure[Scalable]{
		\includegraphics[width= 0.225\linewidth]{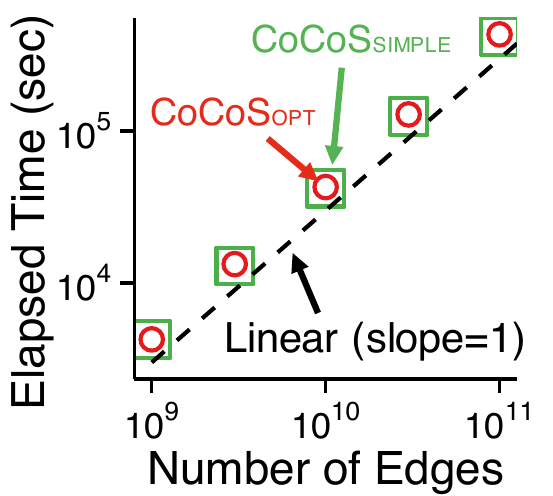}
		\label{fig:trifly:crown:scalability}
	}
	\subfigure[Unbaised with small variance (Theorems~\ref{thm:trifly:bias:method} and \ref{thm:trifly:var:method})]{
		\includegraphics[width= 0.235\linewidth]{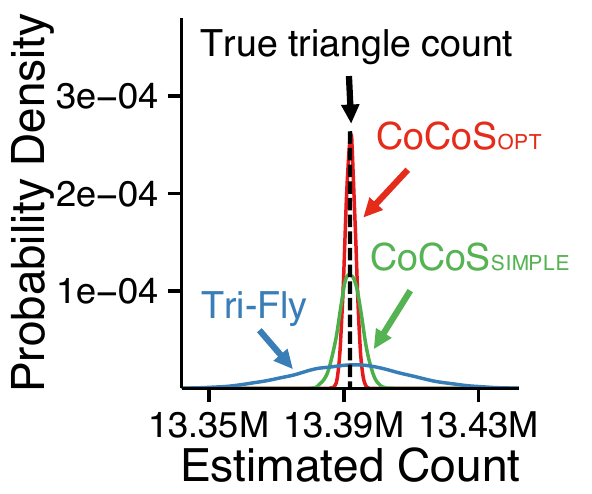}
		\label{fig:trifly:crown:unbias}
		\includegraphics[width= 0.245\linewidth]{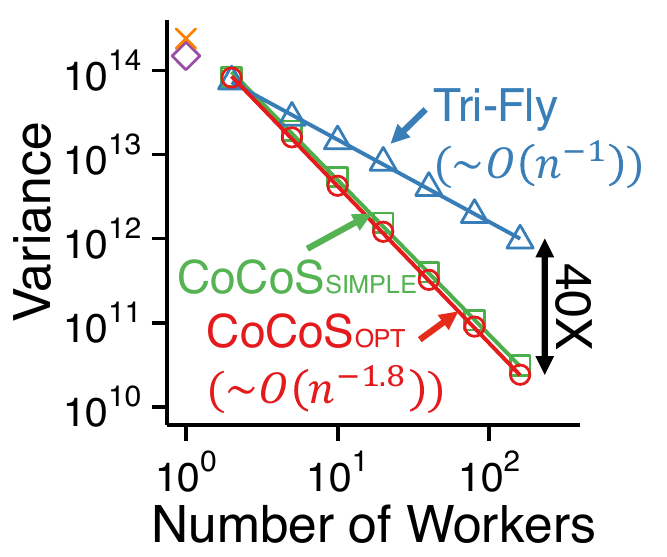}
		\label{fig:trifly:crown:variance}
	} \\
	\vspace{-2mm}
	\caption[Strengths of \cocos]{\label{fig:trifly:crown} \figsummary{Strengths of \cocos.}
	(a) {\it Fast and accurate:} \cocos is faster and more accurate than the baseline method, namely \trifly. (b) {\it Scalable:} The running time of \cocos is linear in the number of edges in the input stream. (c-d) {\it Unbiased with small variance:} \cocos gives unbiased estimates with variances dropping rapidly as we use more machines (Theorem~\ref{thm:trifly:bias:method}). 
	See Section~\ref{sec:trifly:exp} for details.}
\end{figure}

Given a graph stream,
how can we utilize multiple machines for rapidly and accurately estimating the count of triangles in it?
How should we process and sample the edges across the machines to minimize the redundant use of computational and storage resources?

The count of triangles (i.e., cliques of size three) is a computationally expensive but important graph statistic that has proven useful in diverse areas.
For example, the counts of global triangles (i.e., all triangles) and local triangles (i.e., triangles associated with each node) lie at the heart of many crucial concepts in social network analysis and graph theory, including the transitive ratio \cite{newman2003structure}, local clustering coefficients \cite{watts1998collective}, social balance \cite{wasserman1994social}, and trusses \cite{cohen2008trusses}.
The global and local triangle counts have also been used in many data mining and database applications, including link recommendation \cite{tsourakakis2011spectral,epasto2015ego}, anomaly detection \cite{lim2018memory}, spam detection \cite{becchetti2008efficient,becchetti2010efficient}, dense subgraph mining \cite{wang2010triangulation}, degeneracy estimation \cite{shin2018pattern}, and query optimization \cite{bar2002reductions}.

For triangle counting in real-world graphs, many of which are large and evolving with new edges, recent work has focused largely on streaming algorithms \cite{kutzkov2013streaming,lim2018memory,stefani2017triest,pavan2013parallel,ahmed2014graph,ahmed2017sampling,pavan2013counting,shin2017wrs,shin2018think,pagh2012colorful}. 
Given a graph stream, which is a sequence of edges that may not fit in the underlying storage, these algorithms estimate the count of triangles while making a single pass over the stream.
Especially, these algorithms maintain and gradually update their estimates as each edge is received rather than operating on the entire graph.
Thus, they are appropriate for dynamic graphs, whose edges are received over time.

Another popular approach is to extend triangle counting algorithms to distributed settings, including distributed-memory settings \cite{arifuzzaman2013patric} and \mapreduce settings \cite{cohen2009graph,suri2011counting,park2013efficient,park2014mapreduce,park2016pte,park2018enumerating}.
These distributed algorithms utilize computational and storage resources of multiple machines for speed and scalability.
However, unlike streaming algorithms, they require all edges to be given at once.
Thus, they are not applicable to dynamic graphs, whose edges are received over time, or graphs that are too large to fit in the underlying storage.

Can we have the best of both worlds? In other words, can we utilize multiple machines for rapid and accurate triangle counting in a graph stream?
A promising approach is \trifly (see Section~\ref{sec:trifly:method:baseline}),
where edges are broadcast to every machine that independently
runs a state-of-the-art streaming algorithm called \triestimp \cite{stefani2017triest}.
The final estimates are the averages of
the estimates provided by all the machines. 
Although \trifly successfully reduces estimation error inversely proportional to the
number of machines, \trifly incurs a redundant use of
computational and storage resources.

In this work, we propose \cocos ({\bf Co}nditional {\bf Co}unting and {\bf S}ampling), a fast and accurate distributed streaming algorithm that estimates the counts of global and local triangles.
\cocos gives the advantages of both streaming and distributed algorithms, significantly outperforming \trifly, as shown in Figure~\ref{fig:trifly:crown}.
\cocos minimizes the redundant use of computational and storage resources by carefully processing and sampling edges across distributed machines so that each edge is stored in at most two machines and each triangle is counted by at most one machine.
We theoretically and empirically demonstrate that \cocos has the following advantages:
\begin{itemize}[leftmargin=10pt]
	\item {\bf Accurate}: \cocos yields up to $\mathit{30\times}$ and $\mathit{39\times}$ {\it smaller estimation errors} for global and local triangle counts, respectively, than baselines with similar speeds (Figure~\ref{fig:trifly:crown:tradeoff}).
	\item {\bf Fast}: \cocos scales linearly with the number of edges in the input stream (Figure~\ref{fig:trifly:crown:scalability}), and it is up to  $\mathit{10.4\times}$ {\it faster} than baselines while giving more accurate estimates (Figure~\ref{fig:trifly:crown:tradeoff}).
	\item {\bf Theoretically Sound:} \cocos gives unbiased estimates. 
\end{itemize}
{\bf Reproducibility:} The source code and datasets used in this chapter are available at \textit{\url{http://dmlab.kaist.ac.kr/cocos/}}.

This paper is an extended version of \cite{shin2018trifly}, where we proposed \trifly (Section~\ref{sec:trifly:method:baseline}), which we regard as a baseline approach in this paper.
In this extended version, we propose a new algorithm, namely \cocos (Section~\ref{sec:trifly:method:algorithm}), which significantly outperforms \trifly in terms of speed and accuracy, as shown in Figure~\ref{fig:trifly:crown}.
Moreover, we theoretically analyze the accuracy and complexity of \cocos (Section~\ref{sec:trifly:analysis}). In addition, we conduct extensive experiments on $8$ real-world graph datasets to evaluate the efficiency, effectiveness, and scalability of \cocos and the effects of its parameters on the performance (Section~\ref{sec:trifly:exp}).

The rest of this paper is organized as follows.
In Section~\ref{sec:trifly:related}, we review some related studies.
In Section~\ref{sec:trifly:prelim}, we introduce some preliminary concepts, notations, and a formal problem definition. 
In Section~\ref{sec:trifly:method}, we present our proposed algorithm, namely
\cocos, and a baseline algorithm, namely \trifly.
In Section~\ref{sec:trifly:analysis}, we theoretically analyze the accuracy and complexity of them.
After sharing some experimental results in Section~\ref{sec:trifly:exp}, 
we provide conclusions in Section~\ref{sec:trifly:summary}.

\section{Related Work}
\label{sec:trifly:related}
\begin{table}[t]
	\centering
	\small
	\caption{
		\label{tab:comparison} Comparison of triangle counting methods. \cocos satisfies all the criteria while clearly outperforming \trifly in terms of speed and accuracy. 
	}
	\scalebox{0.97}{
		\begin{tabular}{l|c|c|c|c|c|c|c|c}
			\toprule
			& \rotatebox[origin=l]{90}{\cite{tsourakakis2009doulion,tsourakakis2008fast}} &  
			\rotatebox[origin=l]{90}{\cite{becchetti2008efficient,becchetti2010efficient}}&
			\rotatebox[origin=l]{90}{\cite{cohen2009graph,park2018enumerating,suri2011counting}} &
			\rotatebox[origin=l]{90}{\cite{arifuzzaman2013patric,pagh2012colorful}} & 
			\rotatebox[origin=l]{90}{\cite{jha2013space,ahmed2014graph,tangwongsan2013parallel,kallaugher2017hybrid}} &
			\rotatebox[origin=l]{90}{\cite{kutzkov2013streaming,stefani2017triest,lim2018memory,wang2019rept,shin2017wrs,shin2018think}} & 
			\rotatebox[origin=l]{90}{\trifly (Section~\ref{sec:trifly:method:baseline})} &
			\rotatebox[origin=l]{90}{\cocos ~~ (Section~\ref{sec:trifly:method:algorithm})} \\
			\midrule
			Single-Pass Stream Processing & & & & & {\footnotesize \cmark} & {\footnotesize \cmark} & {\bf \large \cmark}  &  {\bf \large \cmark} \\
			Approximation for Large Graphs* & {\footnotesize \cmark} & {\footnotesize \cmark} & & {\footnotesize \cmark} & {\footnotesize \cmark} & {\footnotesize \cmark} & {\bf \large  \cmark} &  {\bf \large \cmark}\\
			Global \& Local Triangle Counting & & {\footnotesize \cmark} & {\footnotesize \cmark} & & & {\footnotesize \cmark} & {\bf \large \cmark} &  {\bf \large \cmark} \\
			Larger Data with More Machines & & & {\footnotesize \cmark} & {\footnotesize \cmark} &  &  & {\bf \large  \cmark} &  {\bf \large \cmark} \\
			More Accurate with More Machines & & & & {\footnotesize \cmark} &  & & {\bf \large  \cmark} &  {\bf \large \cmark}\\
			\bottomrule
			\multicolumn{8}{l}{\small* graphs that are too large to fit in the underlying storage.}
		\end{tabular} 
	}
\end{table}

Numerous algorithms have been developed for triangle counting in many different settings, including shared-memory settings \cite{rahman2013approximate,shun2015multicore,kim2014opt} and external-memory settings \cite{hu2013massive,hu2014efficient,kim2014opt}.
We review related work focusing on streaming algorithms and distributed algorithms for triangle counting. See Table~\ref{tab:comparison} for a summary.

\subsection{Single-Machine Streaming Algorithms}
\label{sec:trifly:related:streaming}
Most streaming algorithms for triangle counting employ sampling for estimation with limited storage.

\noindent{\bf Counting global triangles.} Tsourakakis et al. \cite{tsourakakis2009doulion} proposed sampling each edge independently with equal probability $p$ and then estimating the global triangle count from that in the sampled graph using the fact that each triangle is sampled with probability $p^{3}$.
To increase the probability from $p^{3}$ to $p^{2}$,
Pagh and Tsourakakis \cite{pagh2012colorful} proposed the colorful sampling scheme where each node is colored with a color chosen uniformly at random
among $1/p$ colors and the edges whose endpoints have the same color are stored.
Kallaugher and  Price \cite{kallaugher2017hybrid} proposed sampling each node with equal probability $p$ and storing all edges between the sampled nodes and $p$ of the edges between sampled nodes and unsampled nodes. 
This requires fewer samples than the colorful sampling scheme for the same accuracy guarantee \cite{kallaugher2017hybrid}.
Jha et al. \cite{jha2013space} and Pavan et al. \cite{pavan2013counting} proposed sampling wedges (i.e., paths of length two) in addition to edges; and Ahmed et al. 
\cite{ahmed2014graph,ahmed2017sampling} proposed sampling edges with different probabilities, depending on the counts of adjacent sampled edges and incident triangles. 
Tangwongsan et al.
\cite{tangwongsan2013parallel} proposed a shared-memory, parallel, cache-oblivious version of \cite{pavan2013counting}. However, this parallelization is applicable only when edges arrive in batches rather than one by one.

\noindent{\bf Counting local triangles.} 
The colorful sampling scheme \cite{pagh2012colorful}, described in the previous paragraph, was applied to local triangle counting \cite{kutzkov2013streaming}.
Lim et al. \cite{lim2018memory} proposed \mascot, which uses simple uniform edge sampling but updates its estimates whenever an edge arrives even if it is not sampled.
Wang et al. \cite{wang2019rept} proposed \rept, which is a parallel version of \mascot in multi-core settings. Each processor maintains a separate sample of edges, while all processors update their estimates whenever an edge arrives.

De Stefani et al. \cite{stefani2017triest} proposed \triestimp, which uses reservoir sampling to maintain as many sample edges as storage allows.
Shin \cite{shin2017wrs} improved upon \triestimp in terms of accuracy under the assumption that edges are streamed in the order that they are created.
In addition, Becchetti et al. \cite{becchetti2008efficient,becchetti2010efficient} explored semi-streaming algorithms that require multiple passes over the stream.
Moreover, Shin et al. \cite{shin2018think} and De Stefani et al.\cite{stefani2017triest} explored the local triangle counting in a fully dynamic graph stream with both edge insertions and deletions.

Our algorithm adapts \triestimp for triangle counting within each machine since it estimates both global and local triangle counts accurately without any parameter or assumption.
Note that properly setting the parameters of \mascot \cite{lim2018memory} and \rept \cite{wang2019rept} requires the number of edges in the input graph stream, which is rarely known in advance.
However, any single-machine streaming algorithm can be used instead. 
For example, \wrs \cite{shin2017wrs} can be used if edges in the input graph stream are sorted in chronological order, and \rept \cite{wang2019rept} can be used instead of \triestimp. when each machine is equipped with multiple cores.
Moreover, \thinkd \cite{shin2018think} can be used if the input graph stream is fully-dynamic with both edge insertions and deletions. 

\vspace{-2mm}
\subsection{Distributed Batch Algorithms}
\vspace{-1mm}
Cohen \cite{cohen2009graph} proposed the first triangle counting algorithm on \mapreduce, which directly parallelizes a serial algorithm.
Suri and Vassilvitskii \cite{suri2011counting}, Park et al. \cite{park2013efficient,park2014mapreduce,park2016pte,park2018enumerating}, and Arifuzzaman et al. \cite{arifuzzaman2013patric} proposed dividing the input graph into overlapping subgraphs and assigning them to multiple machines, which count the triangles in the assigned subgraphs in parallel, in \mapreduce settings \cite{suri2011counting,park2013efficient,park2014mapreduce,park2016pte} and distributed-memory settings \cite{arifuzzaman2013patric}.
Recently, Ko and Kim \cite{ko2018turbograph++} proposed an external-memory distributed graph analytics system that supports triangle counting.
These distributed algorithms are for exact triangle counting in static graphs, all of whose edges are given at once.
They are not applicable when edges are received over time and the edges may not fit in the underlying storage, as assumed in this work.


\vspace{-2mm}
\subsection{Distributed Streaming Algorithms}
\vspace{-1mm}
Distributed streaming algorithms for triangle counting were first discussed by Pavan et al. \cite{pavan2013parallel} to handle multiple sources.
Their goal, however, was to reduce communication costs while giving the same estimation of their single-machine streaming algorithm \cite{pavan2013counting}.
Thus, using more machines, which are one per source, neither improves the speed nor the accuracy of the estimation.
In this work, however, we utilize multiple machines for faster and more accurate estimation.



\section{Preliminaries and Problem Definition}
\label{sec:trifly:prelim}
In this section, we first introduce some notations and concepts used throughout this paper.
Then, we define the problem of distributed global and local triangle counting in a graph stream.

\begin{table}[t]
	\centering
	\small
	\caption[Table of frequently-used symbols]{\label{tab:trifly:symbols} \figsummary{Table of frequently-used symbols.}}
		\begin{tabular}{r|l}
			\toprule
			\textbf{Symbol} & \textbf{Definition} \\
			\midrule
			\multicolumn{2}{l}{Notations for Graph Streams (Section~ \ref{sec:trifly:prelim})} \\
			\midrule
			$\gstream$ & input graph stream \\
			$\et$ & edge that arrives at time $t\in\{1,2,...\}$\\
			$\pair$ & edge between nodes $u$ and $v$ \\
			$\tuv$ & arrival time of edge $\pair$\\
			$\triple$ & triangle composed of nodes $u$, $v$, and $w$ \\
			$\SGT=(\SVT,\SET)$ & graph at time $t$\\
			$\STT$ & set of global triangles in $\SGT$\\
			$\STTu$ & set of local triangles containing node $u$ in $\SGT$\\
			\midrule
			\multicolumn{2}{l}{Notations for Algorithms (Section~\ref{sec:trifly:method})} \\
			\midrule
			$\workernum$ & number of workers \\
			$\budget$ & maximum number of edges stored in each worker \\
			$\cbar$ & estimate of the global triangle count \\
			$\cu$ & estimate of the local triangle count of node $u$ \\
			$f:\SV \rightarrow \{1,...,k\}$ & function assigning nodes to workers  \\
			$\li$ & load of the $i$-th worker \\ 
			$\theta$ & tolerance threshold for load difference \\
			\midrule
			\multicolumn{2}{l}{Notations for Analysis (Section~\ref{sec:trifly:analysis})} \\
			\midrule
			$\pt$ & number of Type~1 triangle pairs in $\SGT$ \\
			$\qt$ & number of Type~2 triangle pairs in $\SGT$ \\
			\bottomrule
		\end{tabular}
\end{table}


\subsection{Notations and Concepts} 
We list the frequently-used symbols in Table~\ref{tab:trifly:symbols}.
Consider a graph stream $\gstream$, where $\et$ denotes the undirected 
edge that arrives at time $t\in\{1,2,...\}$.
Then, let $\SGT=(\SVT,\SET)$ be the graph composed of the nodes and edges arriving at time $t$ or earlier. 
We use the unordered pair $\pair\in\SET$ to indicate the edge between two distinct nodes $u,v\in\SVT$.
We denote the arrival time of each edge $\pair$ by $\tuv$. 
We use the unordered triple $\triple$ to indicate the triangle (i.e., three nodes every pair of which is connected by an edge) composed of three distinct nodes $u,v,w\in\SVT$.
We let $\STT$ be the set of {\it global triangles} in $\SGT$ (i.e., all triangles in $\SGT$), and for each node $u\in\SVT$, let $\STTu\subseteq\STT$ be the set of {\it local triangles of $u$} in $\SGT$ (i.e., all triangles containing $u$).

\subsection{Problem Definition}
\label{sec:trifly:prelim:problem}
In this work, we consider the problem of estimating the counts of global and local triangles in a graph stream (i.e., a sequence of edges) using multiple machines with limited storage. Specifically, we assume the following realistic conditions:
\ben
\item[C1] {\bf Knowledge free:} No prior knowledge of the input graph stream (e.g., the counts of nodes and edges) is available.
\item[C2] {\bf Shared nothing environment:} Data stored in the storage of a machine is not accessible by the other machines.
\item[C3] {\bf One pass:} Edges are accessed one by one in their arrival order. Past edges are not accessible by a machine unless they are stored in the given storage of the machine.
\een


Under these conditions, we define the problem of distributed estimation of global and local triangle counts in a graph stream.
\begin{problem}[Distributed Estimation of Global and Local Triangle Counts in a Graph Stream\label{problem:trifly}] \ 
\bit
	\item {\bf Given:} 
	 a graph stream $\gstream$ and  $\workernum$ distributed storages in each of which up to $\budget$ $(\geq 2)$ edges can be stored
	\item {\bf Maintain:} estimates of the global triangle count $\globalnum$ and the local triangle counts  $\{(u,\localnum)\}_{u\in\SVT}$ for current time $t\in\ints$,
	\item {\bf to Minimize:} the biases and variances of the estimates.
%
\eit
\end{problem}

There can be multiple ways of measuring estimation error, including those considered in Section~\ref{sec:trifly:experiments:settings}.
Instead of aiming to minimize a specific measure of estimation error, we use a general approach of simultaneously reducing the biases and variances of estimates.
In Section~\ref{sec:trifly:exp}, we evaluate the proposed algorithms using five different measures of estimation error.



\section{Proposed Algorithms: \trifly and \cocos}
\label{sec:trifly:method}
In this section, we present two distributed streaming algorithms for Problem~\ref{problem:trifly}. 
First, we provide an overview with the common structure and notations in Section~\ref{sec:trifly:method:common}.
Then, we present a baseline algorithm \trifly and our proposed algorithm \cocos ({\bf Co}nditional {\bf Co}unting and {\bf S}ampling) in Sections~\ref{sec:trifly:method:baseline} and \ref{sec:trifly:method:algorithm}, respectively.
After that, we discuss lazy aggregation in Section~\ref{sec:trifly:method:lazy}.
Lastly, we discuss extensions of the algorithms with multiple sources, masters, and aggregators in Section~\ref{sec:trifly:method:multiple}

\begin{figure}
	\centering
	\includegraphics[width= 0.8\linewidth]{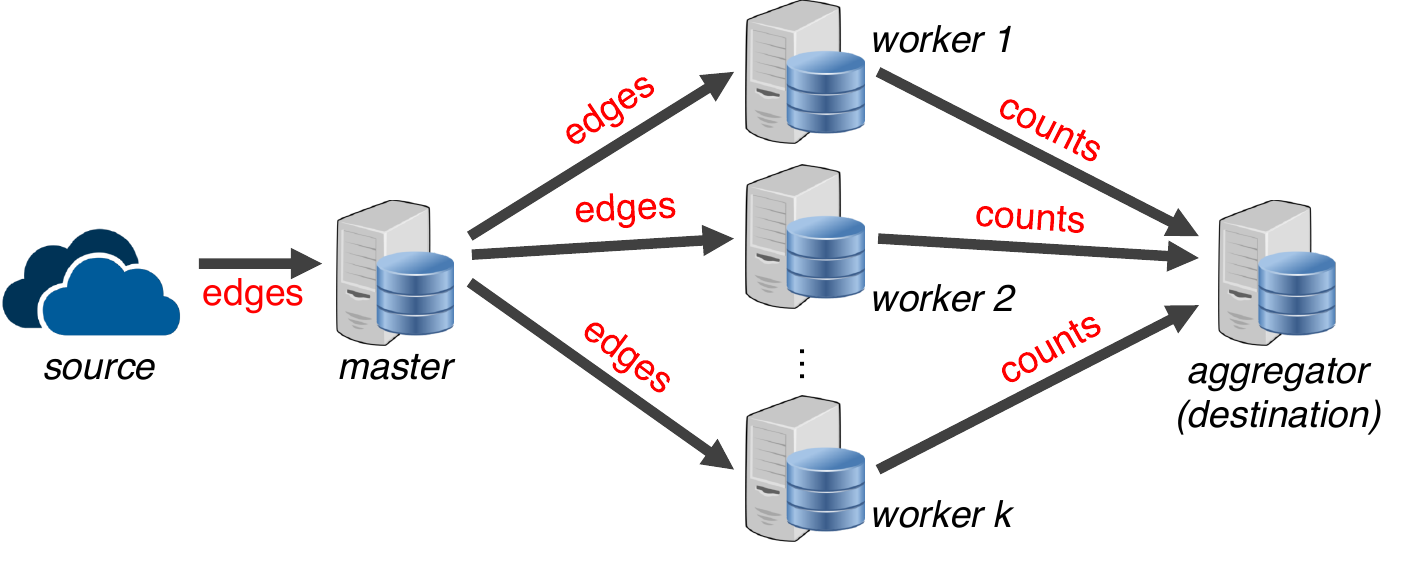}
	\caption[Roles of machines and the flow of data in \cocos]{\label{fig:structure} \figsummary{Roles of machines and the flow of data in \trifly and \cocos.} Extensions of them with multiple sources, masters, and aggregators are discussed in Section~\ref{sec:trifly:method:multiple}.}
\end{figure}

\subsection{Overview}
\label{sec:trifly:method:common}

Figure~\ref{fig:structure} describes the roles of machines and the flow of data in the algorithms described in the following subsections.
For simplicity, we assume one source, one master, and one aggregator although 
\trifly and \cocos in general (except for \cocosopt, i.e., the optimized version of \cocos described in Section~\ref{sec:trifly:method:opt}) are easily extended to multiple sources, masters, and aggregators, as discussed in Section~\ref{sec:trifly:method:multiple}.
Edges are streamed from the source to the master, which unicasts or broadcasts the edges to the workers. 
Each worker counts the global and local triangles from the received edges using its local storage, and it sends the counts to the aggregator.
Since we assume a shared-nothing environment in Section~\ref{sec:trifly:prelim:problem}, each worker cannot access data stored in the other workers.
The counts are aggregated in the aggregator, which gives the final estimates of the global and local triangle counts.

Before describing the algorithms, we define the notations used in them. We use $\workernum$ to denote the number of workers and use $\budget$ to denote the storage budget per worker (i.e., the maximum number of edges that we store in each worker). 
For each $i\in\workerset$, we let $\SEI$ be the edges currently stored in the $i$-th worker and let $\SGI=(\SVI,\SEI)$ be the graph composed of the edges in $\SEI$. 
For each node $u\in \SVI$, $\SNI[u]$ denotes the neighboring nodes of $u$ in $\SGI$.
Since its storage is limited, each worker uses sampling to decide which edges to store.
We use $\li$ to denote the number of edges that the $i$-th worker has considered for sampling so far.
Lastly, $\cbar$ indicates the estimate of the global triangle count, and for each node $u$, $\cu$ indicates the estimate of the local triangle count of $u$.

\begin{algorithm}
	\caption{\trifly: Baseline Algorithm}\label{alg:trifly}
	\begin{algorithmic}[1]
		\Require (1) input graph stream: $\gstream$
		\Statex \ \ \quad (2) storage budget in each worker: $\budget$ $(\geq 2)$
		\Ensure (1) estimated global triangle count: $\cbar$
		\Statex	\ \qquad (2) estimated local triangle counts: $\cu$ for each node $u$
		\vspace{1mm}
		\Statex \hspace{-6mm} {\bf Master}:
		\vspace{1mm}
		\For{each edge $\uv$ from the source} \label{alg:trifly:master:start}
		\State	broadcast $\uv$ to every worker \label{alg:trifly:master:end}
		\EndFor
		\vspace{1mm}
		\Statex \hspace{-6mm} {\bf Worker} (each worker with index $i$):
		\vspace{1mm}
		\State $\SEI \leftarrow \emptyset$; $\li \leftarrow 0$ 	\label{alg:trifly:worker:start}
		\For{each edge $\uv$ from the master \label{alg:trifly:worker:receive}} 
		\State \textsc{count}($\uv$)  \label{alg:trifly:worker:count:call}
		\State \textsc{sample}($\uv$) \label{alg:trifly:worker:sample:call}
		\EndFor
		\Procedure{count}{$\uv$}:
		\State $sum \leftarrow 0$ \label{alg:trifly:worker:count:start}
		\For{each node $w\in \SNI[u]\cap\SNI[v]$}  \label{alg:trifly:worker:intersect}
		\State send $(w, 1/(\piuvw))$ to the aggregator \label{alg:trifly:worker:send:one}
		\State $sum\leftarrow sum+1/(\piuvw)$ \label{alg:trifly:worker:send:sum} \Comment{see Eq.~\eqref{eq:trifly:prob:discover} for $\piuvw$}
		\EndFor
		\State send $(*,sum)$, $(u, sum)$ and $(v, sum)$ to the aggregator \label{alg:trifly:worker:send:two} \label{alg:trifly:worker:count:end} 
		\Statex
		\Comment{`$*$' denotes the global triangle count}
		\EndProcedure
		\Procedure{sample}{$\uv$}:
		\State $\li \leftarrow \li +1 $. \label{alg:trifly:worker:sample:start}
		\If{$|\SEI| < \budget$}
		$\SEI \leftarrow \SEI \cup \{\uv\}$ \label{alg:trifly:worker:sample:always}
		\Else \ with probability $\budget/\li$, replace a uniform random edge in $\SEI$ with $\uv$	\label{alg:trifly:worker:sample:conditional:start}
		\label{alg:trifly:worker:sample:conditional:end} \label{alg:trifly:worker:sample:end} \label{alg:trifly:worker:end}
		\EndIf	
		\EndProcedure
		\vspace{1mm}
		\Statex \hspace{-6mm} {\bf Aggregator}: 
		\vspace{1mm}
		\State $\cbar\leftarrow 0$ \label{alg:trifly:aggregator:start}
		\State initialize an empty map $c$ with default value $0$
		\For{each pair ($u$, $\delta$) from the workers}
		\If{$u=*$} 
		$\cbar \leftarrow \cbar + \delta/\workernum$ \label{alg:trifly:aggregator:increase:one}
		\Else \ $\cu \leftarrow \cu + \delta/\workernum$ \label{alg:trifly:aggregator:increase:two}
		\EndIf
		\label{alg:trifly:aggregator:end}
		\EndFor
	\end{algorithmic}
\end{algorithm}

\subsection{Baseline Algorithm: \trifly}
\label{sec:trifly:method:baseline}

We present \trifly, a baseline algorithm for Problem~\ref{problem:trifly}. 
A pseudo code of \trifly is given in Algorithm~\ref{alg:trifly}.
We first describe the master, the workers, and the aggregator in \trifly. 
Then, we discuss its advantages and disadvantages.

\smallsectionnoline{Master (lines~\ref{alg:trifly:master:start}-\ref{alg:trifly:master:end})}:
The master simply broadcasts every edge from the source to every worker. 

\smallsectionnoline{Workers (lines~\ref{alg:trifly:worker:start}-\ref{alg:trifly:worker:end})}:
Each worker independently estimates the global and local triangle counts using \triestimp, a state-of-the-art streaming algorithm based on reservoir sampling.
Note that the workers use different random seeds and thus give different results.
Each worker $i\in\workerset$ starts with an empty storage (i.e., $\SEI=\emptyset$) (line~\ref{alg:trifly:worker:start} of Algorithm~\ref{alg:trifly}).
Whenever it receives an edge $\uv$ (line~\ref{alg:trifly:worker:receive}) from the master, the worker first counts the triangles with $\uv$ in its local storage by calling the procedure \textsc{count} (line~\ref{alg:trifly:worker:count:call}),
Then, the worker calls procedure \textsc{sample} (line~\ref{alg:trifly:worker:sample:call}) to store $\uv$ in its local storage with non-zero probability.
We describe the procedures \textsc{sample} and \textsc{count} below.

In the procedure \textsc{sample} (lines~\ref{alg:trifly:worker:sample:start}-\ref{alg:trifly:worker:sample:end}), each worker $i\in\workerset$ first increases $\li$, the number of edges considered for sampling, by one since the new edge $\uv$ is being considered.
If its local storage is not full (i.e., $|\SEI|<\budget$), the worker stores $\uv$ by adding $\uv$ to $\SEI$ (line~\ref{alg:trifly:worker:sample:always}).
If the local storage is full (i.e., $|\SEI|=\budget$), the worker stores $\uv$ with probability $\budget/\li$ by replacing an edge chosen uniformly at random  in $\SEI$ with $\uv$ (lines~\ref{alg:trifly:worker:sample:conditional:start}-\ref{alg:trifly:worker:sample:conditional:end}).
This is the standard reservoir sampling, which guarantees that each of the $\li$ edges is sampled and included in $\SEI$ with the equal probability $\min(1,\budget/\li)$.

In the procedure \textsc{count} (lines~\ref{alg:trifly:worker:count:start}-\ref{alg:trifly:worker:count:end}),
each worker $i\in\workerset$ finds the common neighbors of nodes $u$ and $v$ in graph $\SGI$, a graph consisting of the edges $\SEI$ in its local storage (line~\ref{alg:trifly:worker:intersect}).
Each common neighbor $w$ indicates the existence of triangle $\triple$.
Thus, for each common neighbor $w$, the worker increases the global triangle count, and the local triangle counts of nodes $u$, $v$, and $w$ by sending the increases to the aggregator (lines~\ref{alg:trifly:worker:send:one} and \ref{alg:trifly:worker:send:two}).
The amount of increase in the counts is $1/(\piuvw)$ for each triangle $\triple$, where
\begin{equation}
	\piuvw:=\min\left(1,\frac{\budget(\budget-1)}{\li(\li-1)}\right) \label{eq:trifly:prob:discover}
\end{equation}
is the probability that triangle $\triple$ is discovered by worker $i$.
In other words, $\piuvw$ is the probability that both $\vw$ and $\wu$ are in $\SEI$ when $\uv$ arrives at worker $i$.\footnote{
	For $\vw$ to be in $\SEI$, $\vw$ should be one among $\budget$ edges sampled from $\li$ edges, i.e., $p[\vw\in\SEI]=\min(1,\budget/\li)$.
		For $\wu$ to be in $\SEI$, given $\vw$ is in $\SEI$,  $\wu$ should be one among $\budget-1$ edges sampled from $\li-1$ edges, i.e., $p[\wu\in\SEI|\vw\in\SEI]=\min(1,(\budget-1)/(\li-1))$.
		Eq.~\eqref{eq:trifly:prob:discover} follows from $\piuvw=p[\wu\in\SEI, \vw\in\SEI]=p[\vw\in\SEI]\times p[\wu\in\SEI|\vw\in\SEI].$}
Increasing counts by $1/(\piuvw)$ guarantees that the expected amount of the increase sent from each worker is exactly $1(=p_{i}[uvw]\times1/(p_{i}[uvw])+(1-p_{i}[uvw])\times0)$ for each triangle, enabling \trifly to give unbiased estimates. See Theorem~\ref{thm:trifly:bias:naive} in Section~\ref{sec:trifly:analysis:accuracy} for a detailed proof.

\smallsectionnoline{Aggregator (lines~\ref{alg:trifly:aggregator:start}-\ref{alg:trifly:aggregator:end})}:
The aggregator maintains and updates the estimate $\cbar$ of the global triangle count and the estimate $\cu$ of the local triangle count of each node $u$.
Specifically, it increases the estimates by $1/\workernum$ of what it receives, averaging the increases sent from the workers (lines~\ref{alg:trifly:aggregator:increase:one} and \ref{alg:trifly:aggregator:increase:two}).

\smallsectionnoline{Advantages and Disadvantages of \trifly}:
Our theoretical and empirical analyses in the following sections show the advantages of \trifly. Specifically, \trifly gives unbiased estimates, and the variances of the estimates decrease inversely proportional to the number of workers
(see Theorems~\ref{thm:trifly:bias:naive} and \ref{thm:trifly:var:naive} in Section~\ref{sec:trifly:analysis:accuracy}).
Moreover, \trifly gives the same results as \triestimp \cite{stefani2017triest}, a state-of-the-art streaming algorithm, when a single worker is used.

However, \trifly incurs a redundant use of computational and storage resources.
Specifically, each edge can be replicated and stored in up to $\workernum$ workers, and each triangle can be counted repeatedly by up to $\workernum$ workers.
Due to its redundant use of storage, no matter how many workers are used, \trifly cannot guarantee exact triangle counts if the number of edges so far (i.e., $t$) is greater than $\budget+1$.


\begin{algorithm}[t]
	\caption{\cocos: Proposed Algorithm}\label{alg:cocos}
	\begin{algorithmic}[1] 
		\Require (1) input graph stream: $\gstream$
		\Statex \ \ \quad (2) storage budget in each worker: $\budget$ $(\geq 2)$
		\Ensure (1) estimated global triangle count: $\cbar$
		\Statex	\  \qquad (2) estimated local triangle counts: $\cu$ for each node $u$
		\vspace{1mm}
		\Statex \hspace{-6mm} {\bf Master}:
		\vspace{1mm}
		\For{each edge $\uv$ from the source} \label{alg:cocos:master:start}
		\If{$f(u)=f(v)$} 
		send $\uv$ to worker $f(u)$  \Comment{Case~\lucky} \label{alg:cocos:master:case1}
		\Else \ send $\uv$ to every worker \Comment{Case~\unlucky} \label{alg:cocos:master:case2}
		\label{alg:cocos:master:end}
		\EndIf
		\EndFor
		\vspace{1mm}
		\Statex \hspace{-6mm} {\bf Worker} (each worker with index $i$):
		\vspace{1mm}
		\State $\SEI \leftarrow \emptyset$; $\li \leftarrow 0$ \label{alg:cocos:worker:start}
		\For{each edge $\uv$ from the master \label{alg:cocos:worker:receive}}
		\State \textsc{count}($\uv$) \Comment{see Algorithm~\ref{alg:trifly} for \textsc{count()}} \label{alg:cocos:worker:count:call}
		\If{$f(u)=i$ or $f(v)=i$}  \Comment{Case~\assign} \label{alg:cocos:worker:case1}
		\State \textsc{sample}($\uv$) \Comment{see Algorithm~\ref{alg:trifly} for \textsc{sample()}} \label{alg:cocos:worker:sample:call} \label{alg:cocos:worker:end}
		\EndIf
		\EndFor
		\vspace{1mm}
		\Statex \hspace{-6mm} {\bf Aggregator}: 
		\vspace{1mm}
		\State $\cbar\leftarrow 0$ \label{alg:cocos:aggregator:start}
		\State initialize an empty map $c$ with default value $0$
		\For{each pair ($u$, $\delta$) from the workers}
		\If{$u=*$} 
		$\cbar \leftarrow \cbar + \delta$ \label{alg:cocos:aggregator:increase:one}
		\Else \ $\cu \leftarrow \cu + \delta$ \label{alg:cocos:aggregator:increase:two}
		\EndIf \label{alg:cocos:aggregator:end}
		\EndFor
	\end{algorithmic}
\end{algorithm}

\subsection{Proposed Algorithm: \cocos}
\label{sec:trifly:method:algorithm}

To address the drawbacks of \trifly, we propose \cocos, an improved algorithm for Problem~\ref{problem:trifly}.
We first provide the main idea behind \cocos. Then, we describe the master, the workers, and the aggregator in \cocos in detail.
After that, we prove its properties.
Lastly, we discuss adaptive node mapping.

\subsubsection{Main Idea}
\label{sec:trifly:method:algorithm:goal} 

When designing \cocos, we aim to minimize the redundant use of computational and storage resources for rapid and accurate estimation of global and local triangle counts. 
Specifically, we design \cocos so that it distributes and stores edges across workers while satisfying the following desirable properties:
\begin{enumerate}
	\item[P1] {\bf Limited Redundancy in Storage}: Each edge is stored in at most two workers.
	\item[P2] {\bf No Redundancy in Computation}: Each triangle is counted by at most one worker.
	\item[P3] {\bf No Disintegrated Triangles}: For each triangle, at least one worker receives all three edges of the triangle, and with non-zero probability, the worker stores both the first and second edges of the triangle when the last edge of the triangle arrives.
\end{enumerate}
P1 is desirable for accuracy. Less redundancy in storage enables us to store more unique edges from which we can estimate triangle counts more accurately.
P2 is desirable for speed. 
P3 is necessary for \cocos to give (almost) exact estimates when storage is (almost) enough.
P3 is what we aim not to compromise while reducing the redundancy in storage and computation. 
For example, further reducing redundancy in storage by 
storing each edge in at most one worker compromises P3 unless all edges are stored in the same worker.\footnote{Consider a graph stream where the edges of a chain graph of indefinite length arrive first and then some other edges arrive. It is not known in advance which edges will arrive later, as stated in Section~\ref{sec:trifly:prelim:problem}. If we can store each edge in at most one worker, in order to guarantee P3, we have no choice but to store all the edges of the chain graph in the same worker. Assume that two edges of the chain graph are stored in different workers. Then, there always exist three nodes $u$, $v$, and $w$ where $\uv$ and $\vw$ are stored in different workers. If $\wu$ arrives after all edges of the chain graph arrive, P3 does not hold for the triangle $\triple$. 
}

\subsubsection{Algorithm Description}
\label{sec:trifly:method:algorithm:detail} \
A pseudo code of \cocos is given in Algorithm~\ref{alg:cocos}.

\smallsectionnoline{Master (lines~\ref{alg:cocos:master:start}-\ref{alg:cocos:master:end})}:
The master requires a function $f$ that maps each node to a worker. We assume that $f$ is given and discuss it later in Section~\ref{sec:trifly:method:opt}. The master sends each edge $\pair$ to the workers depending on $f(u)$ and $f(v)$ as follows:
\bit
	\item Case \lucky (line~\ref{alg:cocos:master:case1}): If nodes $u$ and $v$ are assigned to the same worker by $f$ (i.e., $f(u)=f(v)$), then the master sends $\pair$ only to the worker (i.e., the $f(u)$-th worker).
	\item Case \unlucky (line~\ref{alg:cocos:master:case2}): Otherwise (i.e., if $f(u)\neq f(v)$), the master sends $\pair$ to every worker.
\eit

Consider triangles where $\uv$ is their last edge closing them.
In the first case (i.e., case \lucky), the worker $f(u)$ ($=f(v)$) has received the other two edges of such a triangle and with non-zero probability stored both (see the description of workers below) when $\uv$ arrives.
Thus, sending $\uv$ to the worker $f(u)$ is enough to satisfy P3 in Section~\ref{sec:trifly:method:algorithm:goal}.
In the second case (i.e., case \unlucky), however, neither the worker $f(u)$ nor the worker $f(v)$ can store both the other two edges of such a triangle (see the description of workers below).
Thus, $\uv$ is broadcast so that for each such a triangle $\triple$, the worker $f(w)$ receives $\uv$.
Note that the worker $f(w)$ has received the other two edges (i.e., $\vw$ and $\wu$) and with non-zero probability stored both, and thus P3 is satisfied.

\smallsectionnoline{Workers (lines~\ref{alg:cocos:worker:start}-\ref{alg:cocos:worker:end})}:
The workers start with an empty storage (line~\ref{alg:cocos:worker:start}).
Whenever they receive an edge $\pair$ from the master (line~\ref{alg:cocos:worker:receive}), they count the triangles with $\pair$ in its local storage by calling the procedure \textsc{count} (line~\ref{alg:cocos:worker:count:call}), as in \trifly.
However, the procedure \textsc{sample} is called selectively depending on $f(u)$ and $f(v)$ as follows:
\bit
	\item Case \assign (line~\ref{alg:cocos:worker:case1}): If $f(u)=i$ or $f(v)=i$,  the $i$-th worker considers storing $\pair$ in its local storage by calling \textsc{sample}.
	\item Case \noassign: Otherwise (i.e., if $f(u)\neq i \neq f(v)$), the $i$-th worker simply discards $\pair$ without considering storing it.
\eit

Note that in only one (if $f(u)=f(v)$) or two (if $f(u)\neq f(v)$) workers, the procedure \textsc{sample} is called, and thus $\uv$ is stored with non-zero probability. Thus, P1 in Section~\ref{sec:trifly:method:algorithm:goal} is satisfied.
Recall that, within the procedure \textsc{sample}, $\li$, the number of edges considered for sampling, is increased by one since the new edge $\uv$ is being considered.
Recall that within the procedure \textsc{count}, $1/(\piuvw)$ is computed for each discovered triangle $\triple$. Note that $\li$ is at least two, since $\vw$ and $\wu$ are sampled, and thus the denominator of $\piuvw$ (i.e,. $\li(\li-1)$) cannot be zero. Also note that $\piuvw$ cannot be zero since $b$ is assumed to be at least two.

\smallsectionnoline{Aggregator (lines~\ref{alg:cocos:aggregator:start}-\ref{alg:cocos:aggregator:end}):}
The aggregator applies each received update to the corresponding estimate.
Note that, different from the aggregator in \trifly, the aggregator in \cocos does not divide received updates by the number of workers (i.e., $\workernum$). 
This is because in \cocos, only one worker can count each triangle with non-zero probability, satisfying P2 in Section~\ref{sec:trifly:method:algorithm:goal}, while in \trifly, all $\workernum$ workers can count each triangle with non-zero probability.
We prove this in the following subsection.

\subsubsection{Basic Properties}
\label{sec:trifly:method:algorithm:properties}
\cocos satisfies P1, P2, and P3, which are the desirable properties described in Section~\ref{sec:trifly:method:algorithm:goal}, as stated in Lemma~\ref{lemma:trifly:property}.

\begin{lem}[Properties of \cocos\label{lemma:trifly:property}]	
	Algorithm~\ref{alg:cocos} satisfies P1, P2, and P3.
\end{lem}

\begin{proof}
	First, we prove P1.
	Each edge $\pair$ can be stored in a worker only when case \assign happens. Since case \assign happens in at most two workers (i.e., the $f(u)$-th worker and the $f(v)$-th worker), $\pair$ can be stored in at most two workers.
	Then, we prove P2 and P3 by showing that, for each triangle, there exists exactly one worker that receives all three edges composing the triangle and with non-zero probability stores both the first and second edges when the last edge arrives.
	Consider a triangle $\triple$ and assume $\pair$ is the last edge (i.e., $\tvw<\tuv$ and $\twu<\tuv$) without loss of generality.
	If $f(u)=f(v)$ (case \lucky), none of the workers --- except the $f(u)(=f(v))$-th worker --- can satisfies the condition since $\pair$ is sent only to the $f(u)$-th worker. The $f(u)(=f(v))$-th worker also stores both $\vw$ and $\wu$ with non-zero probability (case \assign happens for both edges).
	If $f(u)\neq f(v)$ (case \unlucky), although $\pair$ is sent to every worker, none of the workers --- except the $f(w)$-th worker --- can store both $\vw$ and $\wu$ (case \noassign happens for at least one of the edges).
	The $f(w)$-th worker, however, stores $\vw$ and $\wu$ with non-zero probability (case \assign happens for both edges). 
	Therefore, in both cases, there exists exactly one worker that receives all three edges composing $\triple$ and with non-zero probability stores both $\vw$ and $\wu$ when $\uv$ arrives.
\end{proof}


\subsubsection{Adaptive Node Mapping Function}
\label{sec:trifly:method:opt}

So far we have assumed that the function $f$, which assigns each node to a worker, is given.
We discuss how to design $f$ and propose \cocosopt, which is \cocos with our proposed function as $f$.
For each node $u\in \SV$, we use $f(u)$ to denote the worker to which $u$ is assigned.
\begin{algorithm}[t]
	\small
	\caption{Master in \cocosopt }\label{alg:cocos:opt}
	\begin{algorithmic}[1] 
		\Require (1) input graph stream: $(\eone,\etwo,...)$
		\Statex \ \ \ \quad (2) tolerance threshold for load difference: $\theta$ $(\geq 0)$ .
		\Ensure edges sent to workers 
		\State $\li \leftarrow 0$, $\forall i\in \workerset$ \Comment{$\li$ denotes the load of each worker $i$}
		\For{each edge $\pair$ from the source}
		\State $\istar\leftarrow \argmin_{i\in \workerset}\li$ \label{alg:opt:minload} \Comment{$\istar$ denotes the worker with the minimum load so far}
		\If{$u$ and $v$ have not been assigned to a worker by $f$} \label{alg:opt:both:start}
		\State $f(u) \leftarrow \istar$; \ \ $f(v) \leftarrow \istar$ \label{alg:opt:both:end} 
		\Comment{$f(x)$ denotes the worker to which a node $x$ is assigned}
		\ElsIf{$u$ has not been assigned to a worker by $f$} \label{alg:opt:one:start}
		\If{$l_{f(v)}\leq (1+\theta)\listar$}  \Comment{If the load difference is below the tolerance threshold $\theta$}
		\State $f(u) \leftarrow f(v)$  \Comment{Reducing redundancy is prioritized}
		\Else 
		\State $f(u) \leftarrow\istar$ \Comment{Otherwise, load balancing is prioritized}
		\EndIf
		\ElsIf{$v$ has not been assigned to a worker by $f$}
		\If{$l_{f(u)}\leq (1+\theta)\listar$} \Comment{If the load difference is below the tolerance threshold $\theta$}
		\State $f(v) \leftarrow f(u)$ \Comment{Reducing redundancy is prioritized}
		\Else 
		\State $f(v) \leftarrow\istar$ \Comment{Otherwise, load balancing is prioritized}
		\label{alg:opt:one:end}
		\EndIf
		\EndIf
		\If{$f(u)=f(v)$}  \label{alg:opt:send:start} \Comment{{Case \lucky}} 
		\State send $\pair$ to worker $f(u)$
		\State $l_{f(u)} \leftarrow l_{f(u)} + 1$ 
		\Else  \Comment{{Case \unlucky}}
		\State send $\pair$ to every worker 
		\State $l_{f(u)} \leftarrow l_{f(u)} + 1$; \ \ $l_{f(v)} \leftarrow l_{f(v)} + 1$ \label{alg:opt:send:end}
		\EndIf
		\EndFor
		\vspace{-1mm}
	\end{algorithmic}
\end{algorithm}

\smallsection{Design Goals}:
We say an edge $\pair$ is assigned to the $i$-th worker if $f(u)=i$ or $f(v)=i$ and thus $\pair$ can possibly be stored in the $i$-th worker. In Algorithm~\ref{alg:cocos}, the load $\li$ of each $i$-th worker denotes the number of edges assigned to the worker. 
Then, two goals that a desirable $f$ function should meet are as follows:
\begin{itemize}
	\item[G1] {\bf Storage:} The redundant use of storage (i.e., the number of edges stored in multiple workers) should be minimized.
	\item[G2] {\bf Load Balancing:}
	A similar number of edges should be assigned to every worker, i.e., $\li\approx\lj$, $\forall i,j\in\workerset$.
\end{itemize}
However, achieving both goals is non-trivial because the goals compete with each other.
For example, in complete graphs, a perfect load balance and thus the second goal are achieved only when the same number of nodes are assigned to each worker. This, however, maximizes the number of edges stored in multiple workers (i.e., ${|\SV| \choose 2} -\sum_{i=1}^{\workernum} {|\SVI| \choose 2}$, where $|\SVI|$ is the number of nodes assigned to the worker $i$), conflicting with the first goal.
On the other hand, in any connected graphs, the redundant use of storage is minimized, and thus the first goal is achieved only when we assign every node to the same worker. However, this maximizes load imbalance, conflicting with the second goal. 
Moreover, due to the conditions in Section~\ref{sec:trifly:prelim:problem}, $f$ should be decided without additional passes or any prior knowledge of the input stream. In \cocos, when a new node arrives, it should be assigned to a worker without any knowledge on future edges.

\smallsection{\cocosopt with Adaptive $f$.}
We propose \cocosopt, where the master, described in Algorithm~\ref{alg:cocos:opt}, adaptively decides the function $f$ based on the current load $\li$ of each worker $i \in\workerset$ so that the redundancy of storage is minimized within a specified level $\theta$ of load difference.

Recall that, in \cocos, case \lucky is preferred over case \unlucky for reducing the redundancy in storage. This is because each edge $\pair$ is stored in at most one worker in case \lucky (i.e., $f(u)=f(v)$), while it is stored in at most two workers in case \unlucky (i.e., $f(u)\neq f(v)$).
Let the $\istar$-th worker be the worker with least assigned edges so far (line~\ref{alg:opt:minload}).
If an edge $\pair$ with two new nodes $u$ and $v$ arrives, the master assigns both nodes to the $\istar$-th worker (lines~\ref{alg:opt:both:start}-\ref{alg:opt:both:end}) for pursuing case \lucky and balancing loads.
If an edge $\pair$ with one new node $u$ (without loss of generality) arrives, the master assigns $u$ to the $f(v)$-th worker, for case \lucky to happen, as long as the load of the $f(v)$-th worker is not higher than $(1+\theta)$ times of the load of the $\istar$-th worker.
Otherwise, load balancing is prioritized, and $u$ is assigned to the $\istar$-th worker (lines~\ref{alg:opt:one:start}-\ref{alg:opt:one:end}).
Once $f(u)$ and $f(v)$ are determined, each edge $\pair$ is sent to the worker(s) depending on $f(u)$ and $f(v)$ as in Algorithm~\ref{alg:cocos}, and the load of the corresponding worker(s) is updated (lines~\ref{alg:opt:send:start}-\ref{alg:opt:send:end}).
Note that $f(u)$ and $f(v)$ are never changed once they are determined,
Since the assignments by $f$ are only in the master, along each edge to each worker, one bit indicating whether the edge is assigned to the worker or not should be sent to be used in line~\ref{alg:cocos:worker:case1} of Algorithm~\ref{alg:cocos}. 

\smallsection{Advantages of \cocosopt}:
By co-optimizing storage and load balancing, \cocosopt stores more unique edges and thus produces more accurate estimates than \cocossimple, which is \cocos using the simple modulo function as $f$.
Although our explanation so far has focused on storage and load balancing,
\cocosopt also improves upon \cocossimple in terms of speed by increasing the chance of case \lucky, which saves not only storage but also communication and computation costs, as summarized in Table~\ref{tab:case:analysis}.

\smallsection{Potential Disadvantages of \cocosopt}:
	Different from the master with a non-adaptive node mapping function $f$ (e.g., a modulo function), the master in \cocosopt should maintain the mapping between all arriving nodes and the workers.
	Thus, the size of required space in the master can increase indefinitely.
	However, in many large-scale real-world graphs (e.g., the Friendster dataset used in Section~\ref{sec:trifly:exp}), the number of nodes is orders of magnitude smaller than that of edges.
	In addition, as described in Section~\ref{sec:trifly:method:multiple}, \cocosopt is not easily extended to multiple sources and masters.

\begin{table}[t]
	\centering
	\small 
	\caption[Advantages of Case~\lucky]{\label{tab:case:analysis} 
		\figsummary{Advantages of Case~\lucky.}
		Case \lucky saves storage, communication, and computation costs, compared to case \unlucky.
	}
	\scalebox{1}{
		\begin{tabular}{l||c|c|c}
			\toprule
			Algorithms & \multicolumn{2}{c|}{\cocos (Proposed)} & \trifly \\
			\midrule
			Cases & \lucky & \unlucky  \\
			\midrule
			storage 
			(edge is stored in at most) &  $1$ worker & $2$ workers & $k$ workers \\
			communication (edge is sent to) & $1$ worker & $\workernum$ workers & $k$ workers \\
			computation (\textsc{count()} is called in) & $1$ worker & $\workernum$ workers & $k$ workers\\
			\bottomrule
		\end{tabular}
	}
\end{table}

\subsection{Lazy Aggregation}
\label{sec:trifly:method:lazy}

In the procedure \textsc{count} of Algorithm~\ref{alg:trifly}, which is commonly used by \trifly and \cocos, each worker sends the update of the local triangle count of node $w$ to the aggregator whenever it discovers each triangle $\triple$ (line~\ref{alg:trifly:worker:send:one}).
Likewise, each worker sends the updates of the global triangle count and the local triangle counts of nodes $u$ and $v$ to the aggregator whenever it processes each edge $\uv$ (line~\ref{alg:trifly:worker:send:two}).
In cases where this eager aggregation is not needed, we reduce the amount of communication by employing lazy aggregation. 
Specifically, counts aggregated locally in each worker are sent to and aggregated in the aggregator (and removed from the workers)  when they are queried.

\subsection{Multiple Sources, Masters and Aggregators}
\label{sec:trifly:method:multiple}

Although our experiments in Section~\ref{sec:trifly:experiments:tradeoff} show that the performance bottlenecks of proposed algorithms are workers rather than the master, multiple masters can be considered for handling multiple sources or for fault tolerance. 
Consider the case when edges are streamed from one or more sources to multiple masters without duplication. 
By simply using the same non-adaptive node mapping function\footnote{Note that a node mapping function $f$ is non-adaptive if its mapping does not depend on any states.} $f$ (e.g., a modulo function) in every master, we can run masters independently without affecting the accuracy of \trifly or \cocos. This is because, in such cases, masters do not have any state and thus have nothing to share with each other.
The mapping function in \cocosopt (i.e., Algorithm~\ref{alg:cocos:opt}) is adaptive since its mapping depends on the loads of workers. \cocosopt is not easily extended to multiple sources and masters since all masters should share their mappings and the loads of workers.

Multiple aggregators are required when outputs (i.e., $1$ global triangle count and $|\SVT|$ local triangle counts) do not fit one machine or aggregation is a performance bottleneck.
In \trifly and \cocos, workers send key-value pairs, whose key is either `$*$' or a node id, to the aggregator (line~\ref{alg:trifly:worker:count:end} of Algorithm~\ref{alg:trifly}). The computation and storage required for aggregation are distributed across multiple aggregators if workers use the same hash function (that maps each key to an aggregator) to decide where to send each key-value pair.

\section{Theoretical Analysis}
\label{sec:trifly:analysis}
We theoretically analyze the accuracy, time complexity, and space complexity of \cocos and \trifly.
Then, based on the results, we provide a guide to setting the parameters of \cocos and \trifly.

\subsection{Accuracy Analysis}
\label{sec:trifly:analysis:accuracy}

We analyze the biases and variances of the estimates given by \cocos and \trifly.
The biases and variances determine the estimation error of the algorithms.
We first prove that both \cocos and \trifly give estimates with no bias.
Then, we analyze the variances of the estimates to give an intuition why \cocos is more accurate than \trifly.

\subsubsection{Bias Analysis}
\label{sec:trifly:analysis:accuracy:bias}

We prove the unbiasedness of \trifly and \cocos. That is, we show that \trifly and \cocos give estimates whose expected values are equal to the true triangle counts.
For proofs, consider $\SGT=(\SVT,\SET)$, which is the graph consisting of the edges arriving at time $t$ or earlier. 
We define $\cbart$ as $\cbar$ in the aggregator after edge $\et$ is processed.
Then, $\cbart$ is an estimate of $\globalnum$, the count of global triangles in $\SGT$.
Likewise, for each node $u\in\SVT$, we define $\ctu$ as $\cu$ in the aggregator after  $\et$ is processed. 
Then, each $\ctu$ is an estimate of $\localnum$, the count of local triangles of $u$ in $\SGT$.

\begin{thm}[Unbiasedness of \trifly\label{thm:trifly:bias:naive}]
	At any time, the expected values of the estimates given by \trifly are equal to the true global and local triangle counts. That is, in Algorithm~\ref{alg:trifly}, 
	\begin{align*}
	& \BE[\cbart]=\globalnum, & \ \forall t\in\ints. \\
	& \BE[\ctu]=\localnum, & \ \forall u\in\SVT, \ \forall t\in\ints.
	\end{align*}
\end{thm}
\begin{proof}
	The unbiasedness of \trifly follows from that of \triestimp \cite{stefani2017triest}, which each worker in \trifly runs independently.
	Let $\cbarit$ be the global triangle count sent from each worker $i$ by time $t$. By line~\ref{alg:trifly:aggregator:increase:one} of Algorithm~\ref{alg:trifly}, $\cbart=\sum_{i=1}^{\workernum}\cbarit/\workernum$.
	From $\BE[\cbarit]=\globalnum$ (Theorem~4.12 of \cite{stefani2017triest}),
	$$\BE[\cbart]=\sum\nolimits_{i=1}^{\workernum}{\BE[\cbarit]}/{\workernum}=\globalnum.$$
	Likewise, for each node $u\in\SVT$, let $\ctiu$ be the local triangle count of $u$ sent from each worker $i$ by time $t$.
	By line~\ref{alg:trifly:aggregator:increase:two} of Algorithm~\ref{alg:trifly},  $\ctu=\sum_{i=1}^{\workernum}\ctiu/\workernum$.
	From $\BE[\ctiu]=\localnum$ (Theorem~4.12 of \cite{stefani2017triest}), 
	$$\BE[\ctu]=\sum\nolimits_{i=1}^{\workernum}{\BE[\ctiu]}/{\workernum}=\localnum.\eqno$$
\end{proof}
\begin{thm}[Unbiasedness of \cocos\label{thm:trifly:bias:method}]  
	At any time, the expected values of the estimates given by \cocos are equal to the true global and local triangle counts. That is, in Algorithm~\ref{alg:cocos}, 
	\begin{align*}
	& \BE[\cbart]=\globalnum, & \ \forall t\in\ints. \\
	& \BE[\ctu]=\localnum, & \ \forall u\in\SVT, \ \forall t\in\ints.
	\end{align*}
\end{thm}

\begin{proof}
Consider a triangle $\triple\in \STT$ and assume without loss of generality that $\tvw < \twu < \tuv\leq t$.
By Lemma~\ref{lemma:trifly:property}, there is exactly one worker that can count $\triple$. Let $\fuvw\in\workerset$ denote the worker. 
Let $\diuvw$ be the contribution of $\triple$ to each of $\cbart$, $\ctu$, $\ctv$, and $\ctw$ by each $i$-th worker.
Then, $\diuvw=0$ if $i\neq \fuvw$.
If we let $\SEfuvw^{(\tuv)}$ be the set of edges stored in the $\fuvw$-th worker when $\uv$ arrives, then by lines~\ref{alg:trifly:worker:send:one}-\ref{alg:trifly:worker:send:two} of Algorithm~\ref{alg:trifly} and lines~\ref{alg:cocos:aggregator:increase:one}-\ref{alg:cocos:aggregator:increase:two} of Algorithm~\ref{alg:cocos},
	\begin{equation*}
	\dfuvwuvw = 
	\begin{cases}
	1/(\pfuvwuvw) & \mbox{if } \vw, \wu \in \SEfuvw^{(\tuv)} \\
	0 & \mbox{ otherwise.}
	\end{cases}
	\end{equation*}
By definition, $\pfuvwuvw$ is the probability that both $\vw$ and $\wu$ are in $\SEfuvw^{(\tuv)}$. Therefore, $\BE[\dfuvwuvw] = 1$. 
By linearity of expectation, the following equations hold:
\begin{align*}
	\BE[\cbart]  & = \BE \bigg[\sum_{i=1}^{\workernum} \sum_{\triple \in \STT} \diuvw \bigg]
	= \sum_{i=1}^{\workernum}\sum_{\triple \in \STT} \BE[\diuvw] \\ & 
	= \sum_{\triple \in \STT} \BE[\dfuvwuvw] = \sum_{{\triple \in \STT}}1 = \globalnum, \qquad \forall t\in\ints.
\end{align*}
\begin{align*}
	\BE[\ctu] & = \BE \bigg[\sum_{i=1}^{\workernum} \sum_{\ \ \ \triple \in \STTu} \diuvw \bigg] = \sum_{i=1}^{\workernum}\sum_{\triple \in \STTu} \BE[\diuvw] \\
	& = \sum_{\triple \in \STTu} \BE[\dfuvwuvw] = \sum_{\triple \in \STTu} 1 \\
	&  = \localnum, \qquad \forall t\in\ints,\forall u\in \SVT.
\end{align*}
Hence, the estimates given by Algorithm~\ref{alg:cocos} are unbiased.
\end{proof}

\subsubsection{Variance Analysis}
\label{sec:trifly:analysis:accuracy:variance}

Having shown that the estimate $\cbart$ is an unbiased estimate of the global triangle count $\globalnum$, we analyze its variance in \trifly and \cocos to give an intuition why the variance is smaller in \cocos than in \trifly.
The variance of each $\ctu$ can be analyzed in the same manner considering only the local triangles with node $u$.
We first define the two types of triangle pairs illustrated in Figure~\ref{fig:trifly:pair:defn}.

\begin{defn}[Type~1 Triangle Pair]
	A Type~1 triangle pair is two different triangles $\triple$ and $\tripletwo$ sharing an edge $\uv$ satisfying $\twu=\max(\tuv,\tvw,\twu)$ and $t_{xu}=\max(\tuv,t_{vx},t_{xu})$.
\end{defn}
\begin{defn}[Type~2 Triangle Pair]
	 A Type~2 triangle pair is two different triangles $\triple$ and $\tripletwo$ sharing an edge $\uv$ satisfying $\tvw=\max(\tuv,\tvw,\twu)$ and $t_{xu}=\max(\tuv,t_{vx},t_{xu})$.
\end{defn}

\begin{figure}[t]
	\centering
	\includegraphics[width=0.7\linewidth]{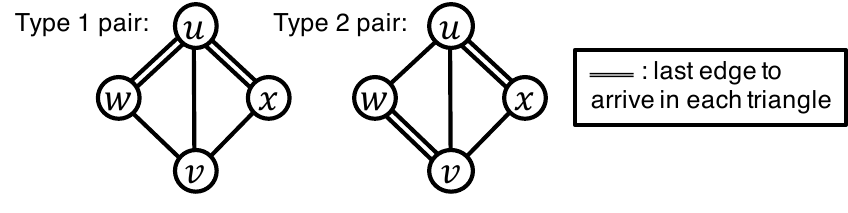}
	\caption[Illustrations of Type~1 and Type~2 triangle pairs]{\label{fig:trifly:pair:defn} \figsummary{Illustrations of Type~1 and Type~2 triangle pairs.}} 
\end{figure}

Let $\pt$ and $\qt$ be the numbers of Type~1 pairs and Type~2 pairs, respectively, in $\SGT$, which is the graph composed of the edges arriving at time $t$ or earlier.
Then, we define $\zt$ as 
$$\zt:=\max\left(0, \globalnum\left(\frac{(t-1)(t-2)}{\budget(\budget-1)}-1\right) + \left(\pt+\qt\right)\frac{t-1-\budget}{\budget}\right), $$
Our analysis in this section is largely based on Lemma~\ref{lemma:trifly:var:naive:one}, where $\zt$ upper bounds the variance of the estimate $\cbart$ in \triestimp, which is equivalent to \trifly and \cocos with a single worker.
Notice that $\zt$ decreases as the storage budget (i.e., $\budget$) increases, while $\zt$ increases as the numbers of edges (i.e., $t$), triangles (i.e., $\globalnum$), and Type~1 or 2 triangle pairs (i.e., $\pt$ and $\qt$) increase.

\begin{lem}[Variance of \triestimp \cite{stefani2017triest}]  \label{lemma:trifly:var:naive:one}
	Assume that a single worker is used (i.e., $\workernum=1$) in Algorithm~\ref{alg:trifly} or Algorithm~\ref{alg:cocos}. At any time $t$, the variance of the estimate $\cbart$ of the global triangle count $\globalnum$ is upper bounded by $\zt$. That is,
	$$Var[\cbart] \leq \zt,\ \forall t\in\ints.$$
\end{lem}

The upper bound of the variance of the estimate $\cbart$ in \trifly decreases proportionally to the number of workers, as formalized in Theorem~\ref{thm:trifly:var:naive}.
	This follows from the fact that $\cbart$ in \trifly is the simple average of $\workernum$ estimates obtained by running $\triestimp$ independently in $\workernum$ workers.

\begin{thm}[Variance of \trifly]  \label{thm:trifly:var:naive}
	In Algorithm~\ref{alg:trifly}, the upper bound of the variance of the estimate $\cbart$, given in Lemma~\ref{lemma:trifly:var:naive:one}, decreases proportionally to the number of workers $\workernum$. That is,
	\begin{equation}
		Var[\cbart] \leq \zt/\workernum,\ \forall t\in\ints. \label{eq:trifly:var:naive}
	\end{equation}
\end{thm}
\begin{proof}
Let $\cbarit$ be the global triangle count sent from each worker $i$ by time $t$. Then, by line~\ref{alg:trifly:aggregator:increase:one} of Algorithm~\ref{alg:trifly}, $\cbart=\sum_{i=1}^{\workernum}\cbarit/\workernum$.
Since $\cbarit$ of each worker $i \in \workerset$ is independent from that of the other workers, 
\begin{equation*}
	Var[\cbart] = \sum\nolimits_{i=1}^{k} Var[\cbarit / \workernum] = \sum\nolimits_{i=1}^{k} Var[\cbarit] / \workernum^2 \leq \workernum\cdot \zt/\workernum^2 = \zt/\workernum,
\end{equation*}
where the inequality follows from Theorem~4.13 in \cite{stefani2017triest}, which states that $Var[\cbarit] \leq \zt$ for each worker $i\in\workerset$.
\end{proof}

The variance of the estimate $\cbart$ in \cocos depends on how the triangles in $\STT$ are distributed across workers.
By Lemma~\ref{lemma:trifly:property}, there is exactly one worker that can count each triangle.
Thus, for each $i\in\workerset$, let $\STTI\subseteq \STT$ be the set of triangles that can be counted by the $i$-th worker.
Likewise, let $\pti$ and $\qti$ be the numbers of Type~1 pairs and Type~2 pairs, respectively, among the triangles in $\STTI$.
Then, for each $i$-th worker, we define $\zti$ as
\begin{equation*}
	\zti:=\max\left(0,\globalnumi \Big(\frac{(\lti-1)(\lti-2)}{\budget(\budget-1)}-1\Big) + \big(\pti+\qti\big)\frac{\lti-1-\budget}{\budget}\right),
\end{equation*}
where $\lti$ is the load $\li$ of each $i$-th worker when $\et$ arrives.
This term is used to upper bound  the variance of $\cbart$ in
Theorem~\ref{thm:trifly:var:method}.
According to the theorem, each worker's contribution to the variance decreases as the storage budget $\budget$ increases, while the contribution increases as more edges, triangles, and Type~1 or 2 triangle pairs (whose discovering probabilities are positively correlated) are assigned to the worker, which matches our intuition.

\begin{thm}[Variance of \cocos] \label{thm:trifly:var:method}
	At any time $t$, the variance of the estimate $\cbart$ of the global triangle count $\globalnum$ in Algorithm~\ref{alg:cocos} with a deterministic node mapping function $f$ is upper bounded by the sum of $\zti$ over all workers. That is
	\begin{equation}
	Var[\cbart] \leq \sum\nolimits_{i=1}^{\workernum}\zti,\ \forall t\in\ints \label{eq:trifly:var:method}
	\end{equation}
\end{thm}
\begin{sproof}
Let $\cbarit$ be the global triangle count sent from each $i$-th worker to the corresponding aggregator by time $t$. Then, by line~\ref{alg:cocos:aggregator:increase:one}, $\cbart = \sum_{i=1}^{\workernum} \cbarit$. 
Since $f$ is assumed to be deterministic and the sampling processes of different workers are independent, $\cbarit$ of each $i$-th worker is uncorrelated with that of the other workers. Thus, Eq.~\eqref{eq:trifly:method:var:proof:one} holds. 
\begin{equation}
	Var[\cbart] = \sum\nolimits_{i=1}^{\workernum} Var[\cbarit]. \label{eq:trifly:method:var:proof:one}
\end{equation}
Then, Theorem~4.13 in \cite{stefani2017triest} is generalized for each $\cbarit$ to $Var[\cbarit]$ $\leq \zti.$
This generalization and Eq.~\eqref{eq:trifly:method:var:proof:one} imply Eq.~\eqref{eq:trifly:var:method}.
\end{sproof}
Note that the adaptive mapping function used in \cocosopt is also deterministic if we break ties in a deterministic way. Given an input graph stream, it always gives the same mapping.

We compare the variance of $\cbart$ in \cocos (i.e., Eq.~\eqref{eq:trifly:var:method}) with that in \trifly (i.e., Eq.~\eqref{eq:trifly:var:naive}). To this end, we say a node mapping function $f:\SV \rightarrow \{1,...,k\}$ is {\it $\epsilon$-uniform} if it satisfies the following conditions for all $i\in\{1,...,k\}$:
\begin{align*}
\globalnumi \leq (1+\epsilon)\cdot \frac{\globalnum}{\workernum}, \ \ \lti \leq (1+\epsilon)\cdot\frac{t}{\workernum},
 \ \ \pti \leq (1+\epsilon)\cdot\frac{\pt}{\workernum},  \ \
\qti \leq (1+\epsilon)\cdot\frac{\qt}{\workernum^2}. 
\end{align*}
Note that $\sum_{i=1}^{k}\qti$ can be strictly small than $\qt$.\footnote{A Type~2 triangle pair is not assigned to any worker if the two triangles are assigned to different workers.} 

If \cocos is equipped with an $\epsilon$-uniform mapping function, for each $\zti$ in Eq.~\eqref{eq:trifly:var:method},
\begin{align*}
\zti &
\leq \globalnumi \cdot \frac{(\lti)^2}{\budget(\budget-1)} + \pti \cdot \frac{\lti}{\budget} + \qti \cdot \frac{\lti}{\budget} \\
& \leq  \frac{(1+\epsilon)^{3}}{\workernum^3} \cdot \globalnum \cdot \frac{t^2}{\budget(\budget-1)} + \frac{(1+\epsilon)^2}{\workernum^2}\cdot\pt\cdot \frac{t}{\budget} + \frac{(1+\epsilon)^2}{\workernum^3}\cdot\qt\cdot \frac{t}{\budget} \\
& = O\left(\frac{(1+\epsilon)^{3}\globalnum t^{2}}{\workernum^{3}\budget^{2}} + \frac{(1+\epsilon)^{2}\pt t}{\workernum^{2}\budget} + \frac{(1+\epsilon)^{2} \qt t}{\workernum^{3}\budget} \right). 
\end{align*}
Then, by Theorem~\ref{thm:trifly:var:method}, the variance of the estimate $\cbart$ in \cocos is
\begin{equation}
Var[\cbart]
= O\left(\frac{(1+\epsilon)^{3}\globalnum t^{2}}{\workernum^{2}\budget^{2}} + \frac{(1+\epsilon)^{2}\pt t}{\workernum \budget} + \frac{(1+\epsilon)^{2}\qt t}{\workernum^{2} \budget} \right). \label{var:approx:method}
\end{equation}
If an $O(1)$-uniform node mapping function is used, then, the variance of the estimate $\cbart$ in\cocos becomes
\begin{equation}
Var[\cbart]
= O\left(\frac{\globalnum t^{2}}{\workernum^{2}\budget^{2}} + \frac{\pt t}{\workernum \budget} + \frac{\qt t}{\workernum^{2} \budget} \right). \label{var:approx:method:constant}
\end{equation}

On the other hand, by Theorem~\ref{thm:trifly:var:naive}, the variance of the estimate in \trifly is
\begin{equation} 
Var[\cbart] \leq \globalnum \cdot \frac{t^2}{\budget(\budget-1)} + \left(\pt+\qt\right)\frac{t}{\budget}   
= O\left(\frac{\globalnum t^{2}}{\workernum \budget^{2}} + \frac{\pt t}{\workernum \budget} + \frac{\qt t}{\workernum \budget} \right). \label{var:approx:naive} 
\end{equation}
Notice how rapidly the variances in \cocos with an $O(1)$-uniform mapping function (Eq.~\eqref{var:approx:method:constant}) and \trifly (Eq.~\eqref{var:approx:naive}) decrease depending on the number of workers (i.e., $\workernum$).
In Eq.~\eqref{var:approx:method:constant}, only the second term is $O(1/\workernum)$ while the other terms
are $O(1/\workernum^2)$. In Eq.~\eqref{var:approx:naive}, however, all the terms are $O(1/\workernum)$.
This analysis gives an intuition why the variance of $\cbart$ in \cocos can be smaller than that in \trifly, especially when many workers are used.
See Section~\ref{sec:trifly:experiments:theorem} for empirical comparison of the variances.

\smallsection{Limitations of Our Analysis:}
	The comparison above is based on the assumption that \cocos is equipped with an $O(1)$-uniform mapping function.
	While the uniform random mapping function is ``expected'' to be $1$-uniform, as formalized in Lemma~\ref{lemma:trifly:expectation}, we can easily find some cases (e.g., star graphs with the center node $u$, where $\max_{i\in\{1,...,k\}}\lti=\ltfu=t$) where there exists no $O(1)$-uniform node mapping function.
	We leave further analysis of the existence and identification (especially under the conditions in Section~\ref{sec:trifly:prelim:problem}) of optimal node mapping functions as future work.
	\begin{lem} \label{lemma:trifly:expectation}
		Assume $f:\SV\rightarrow \{1,...,\workernum\}$ is a random function where $\BP[f(u)=i]=1/\workernum$ for each node $u\in \SV$ and each $i$-th worker.
		Let  $\pt$ and $\qt$ be the counts of Type~1 and Type~2 triangle pairs in $\SGT$.
		Then, the following equations hold for $\epsilon=1$ at any time $t\in\ints$:
		\begin{align}
		\BE_{f}[\globalnumi] \leq (1+\epsilon)\cdot \frac{\globalnum}{\workernum},  \ \ \BE_{f}[\lti] \leq (1+\epsilon)\cdot\frac{t}{\workernum}, \label{eq:trifly:expectation:first} \\
		\ \ \BE_{f}[\pti] \leq (1+\epsilon)\cdot\frac{\pt}{\workernum},  \ \
		\BE_{f}[\qti] \leq (1+\epsilon)\cdot\frac{\qt}{\workernum^2}. 	\label{eq:trifly:expectation:second}
		\end{align}	
	\end{lem}
	\begin{proof} 
		See Appendix~\ref{sec:trifly:appendix:proof}. 
	\end{proof}

\subsection{Complexity Analysis}
\label{sec:trifly:analysis:complexity}

We discuss the time and space complexities of \trifly, \cocossimple (\cocos with the simple modulo function as $f$) and \cocosopt (\cocos with Algorithm~\ref{alg:cocos:opt} as $f$).
We assume that sampled edges are stored in the adjacency list format in memory, as in our implementation used in the experiment section.

\begin{table}
	\centering
	\caption[Time and space complexities of \cocos]{\label{tab:complexity} \figsummary{Time and space complexities of processing first $t$ edges in the input stream.} $\sbar:=\min(t,\budget\workernum)\leq \lbar:=\min(t \workernum,\budget\workernum)$.
	}
	\begin{tabular}{l|c|c|c}
		\multicolumn{4}{l}{\bf Time Complexity} \\
		\toprule
		Methods & Master & Workers (Total) & Aggregator \\
		\midrule
		\cocos (both) & \multirow{2}{*}{$O(t\workernum)$*} &  $O(t \sbar)$ &  $O(\min(t \sbar,\globalnum))$* \\
		\trifly &  & $O(t \lbar)$ & $O(\min(t \lbar,\globalnum\cdot\workernum))$* \\
		\bottomrule
		\multicolumn{4}{l}{\ } \\
		\multicolumn{4}{l}{\bf Space Complexity} \\
		\toprule
		Methods & Master & Workers (Total) & Aggregator \\
		\midrule
		\cocossimple & $O(\workernum)$ & $O(\sbar)$ & \multirow{3}{*}{$O(|\SVT|)$*} \\
		\cocosopt & $O(|\SVT|+\workernum))$ & $O(\sbar)$ &  \\
		\trifly & $O(\workernum)$ & $O(\lbar)$  &  \\
		\bottomrule
		\multicolumn{4}{l}{*can be distributed across multiple masters or aggregators  (see Section~\ref{sec:trifly:method:multiple})}
	\end{tabular}
\end{table}

\subsubsection{Time Complexity Analysis}
\label{sec:trifly:analysis:complexity:time}

The time complexities of the considered algorithms for processing $t$ edges in the input stream are summarized in Table~\ref{tab:complexity}.
The master commonly takes $O(t\cdot\workernum)$ since, in the worst case, every edge is broadcast.

The workers in \trifly take $O(t\cdot \min(t \workernum,\budget\workernum))$ in total, while the workers in \cocos take only $O(t\cdot \min(t,\budget\workernum))$ in total, as shown in Theorems~\ref{theorem:trifly:complexity:time:naive} and \ref{theorem:trifly:complexity:time:method}, which are based on Lemma~\ref{lemma:trifly:complexity:time}.
\begin{lem} \label{lemma:trifly:complexity:time}
	Let $\lsi$ be the load $\li$ of the $i$-th worker when $\es$ arrives.
	If the $i$-th worker receives $\es$, then it takes $O(\min(\lsi,\budget))$ to process $\es$ (i.e., to run lines~\ref{alg:trifly:worker:count:call}-\ref{alg:trifly:worker:sample:call} of Algorithm~\ref{alg:trifly} and lines~\ref{alg:cocos:worker:count:call}-\ref{alg:cocos:worker:sample:call} of Algorithm~\ref{alg:cocos}).
\end{lem}
\begin{proof}
The most expensive step of processing $\es=\pair$ in both Algorithms~\ref{alg:trifly} and \ref{alg:cocos} is to find the common neighbors of nodes $u$ and $v$ (line~\ref{alg:trifly:worker:intersect} of Algorithm~\ref{alg:trifly}).
Computing $\SNI[u]\cap\SNI[v]$ requires accessing $|\SNI[u]|+|\SNI[v]|=O(|\SESI|)=O(\min(\lsi,\budget))$ edges, where $\SESI$ is the set of edges stored in the $i$-th worker when $\es$ arrives. 
\end{proof}
\begin{thm}[Time Complexity of Workers in \trifly]
	\label{theorem:trifly:complexity:time:naive}
	In Algorithm~\ref{alg:trifly}, the total time complexity of the workers for processing the first $t$ edges in the input stream is  $O(t\cdot \min(t\workernum,\budget\workernum))$.
\end{thm}
\begin{proof}
	From Lemma~\ref{lemma:trifly:complexity:time},
	processing an edge $\es$ by the workers takes $O(\sum_{i=1}^{\workernum}\min(\lsi,\budget))$ in total.
	Thus, processing the first $t$ edges takes $O\left(\sum_{s=1}^{t} \sum_{i=1}^{\workernum}\min(\lsi,\budget)\right)$. 
	Since $\lsi=s-1$ in Algorithm~\ref{alg:trifly},
	\begin{align*}
	\sum\nolimits_{s=1}^{t} \sum\nolimits_{i=1}^{\workernum}\min(\lsi,\budget) & = \sum\nolimits_{s=1}^{t}\sum\nolimits_{i=1}^{\workernum} \min(s-1,\budget) = \sum\nolimits_{s=1}^{t} \min((s-1)\workernum,\budget\workernum) \\ & \leq t\cdot \min(t\workernum,\budget\workernum).
	\end{align*}
	Hence, the workers take $O(t \cdot \min(t\workernum,\budget\workernum))$ in total to process the first $t$ edges in the input stream.
\end{proof}
\begin{thm}[Time Complexity of Workers in \cocos] \label{theorem:trifly:complexity:time:method}
	In Algorithm~\ref{alg:cocos}, the total time complexity of the workers for processing the first $t$ edges in the input stream is  $O(t\cdot \min(t,\budget\workernum))$.
\end{thm}
\begin{proof}
From Lemma~\ref{lemma:trifly:complexity:time},
processing an edge $\es$ by the workers takes $O(\sum_{i=1}^{\workernum}\min(\lsi,\budget))$ in total.
Thus, processing the first $t$ edges takes $O\left(\sum_{s=1}^{t} \sum_{i=1}^{\workernum}\min(\lsi,\budget)\right)$. 
Since each edge is assigned to at most two workers (i.e.,
P1 in Lemma~\ref{lemma:trifly:property}), $\sum_{i=1}^{\workernum}\lsi \leq 2(s-1)$ holds, and it implies
\begin{align*}
	\sum\nolimits_{s=1}^{t} \sum\nolimits_{i=1}^{\workernum}\min(\lsi,\budget) & \leq \sum\nolimits_{s=1}^{t} \min(\sum\nolimits_{i=1}^{\workernum}\lsi,\sum\nolimits_{i=1}^{\workernum}\budget))\\ & \leq \sum\nolimits_{s=1}^{t} \min(2(s-1),\budget\workernum)) \leq t\cdot \min(2t,\budget\workernum).
\end{align*}
Hence, the workers take $O(t \cdot \min(t,\budget\workernum))$ in total to process the first $t$ edges in the input stream.
\end{proof}

The aggregator takes $O(\globalnum\cdot\workernum)$ in \trifly since, in the worst case, each triangle is counted by every worker and thus the increases in counts by each triangle are sent to the aggregator $\workernum$ times.
In \cocossimple and \cocosopt, however, the aggregator takes $O(\min(|\STT|, t\cdot \min(t,\budget\workernum)))$.
Since the aggregator takes $O(1)$ for each update that it receives, its time complexity is proportional to the number of triangles counted by the workers.
The number of counted triangles is $O(t\cdot \min(t,\budget\workernum))$ by Theorem~\ref{theorem:trifly:complexity:time:method}, and it is $O(|\STT|)$ since each triangle is counted by at most one worker (i.e., P2 in Lemma~\ref{lemma:trifly:property}).
However, the computational cost of the aggregator can be easily distributed across multiple aggregators, as discussed in Section~\ref{sec:trifly:method:multiple}.

Notice that, with a fixed storage budget $\budget$, the time complexities of \cocossimple and \cocosopt are linear in the number of edges in the input stream, as also shown empirically in Section~\ref{sec:trifly:experiments:scalable}.

\subsubsection{Space Complexity Analysis}
\label{sec:trifly:analysis:complexity:space}

The space complexities of the considered algorithms for processing $t$ edges in the input stream are summarized in Table~\ref{tab:complexity}.
In \trifly and \cocossimple, the master requires $O(\workernum)$ space to maintain the addresses of all the workers. 
In \cocosopt, the master requires additional $O(\workernum+|\SVT|)$ space to store the loads of the workers and the mapping between the nodes and the workers (i.e., function $f$) while processing the first $t$ edges in the input stream.

In all the algorithms,
the workers require $O(\sum_{i=1}^{\workernum} \min(\ltip,\budget))$ space in total, to store sampled edges, where $\lti$ is the load $\li$ of the $i$-th worker when $\et$ arrives.
In \trifly, since $\ltip=t$, the space complexity of the workers is $O(\min(t\workernum,\budget\workernum))$ in total.
In \cocossimple ad \cocosopt,
since each edge is stored in at most two workers (i.e., P1 in Lemma~\ref{lemma:trifly:property}), $\sum_{i=1}^{\workernum}\ltip \leq 2t$ holds, and it implies  
\begin{equation*} \sum\nolimits_{i=1}^{\workernum} \min(\ltip,\budget) \leq
	\min(\sum\nolimits_{i=1}^{\workernum}\ltip,\sum\nolimits_{i=1}^{\workernum}\budget) \leq \min(2t,\budget\workernum). 
\end{equation*}
Hence, the total space complexity of the workers is $O(\min(t,\budget\workernum))$. 

In all the algorithms, the aggregator maintains one estimate of the global triangle count and $O(|\SVT|)$ estimates of the local triangle counts. However, this requirement can be easily distributed across multiple aggregators, as discussed in Section~\ref{sec:trifly:method:multiple}.

\subsubsection{A Guide to Setting Parameters}
\label{sec:trifly:analysis:guide}

In this section, we provide a guide to setting the parameters of \cocos and \trifly.
As shown in Sections~\ref{sec:trifly:analysis:accuracy} and \ref{sec:trifly:analysis:complexity}, both the number of workers (i.e., $\workernum$) and the storage budget per worker (i.e., $\budget$) affect the accuracy and speed of \cocos and \trifly.
Which one should we increase first for rapid and accurate estimation?
For example, which one should we choose between $10$ workers with $10GB$ storage each and $100$ workers with $1GB$ storage each?

When using \cocos, $100$ workers with $1GB$ storage each is preferred. That is, we recommend increasing the number of workers (i.e., $\workernum$) first rather than the storage budget per worker (i.e., $\budget$).
As shown in Table~\ref{tab:complexity}, when $t$ is large enough, the elapsed time of \cocos increases linearly with both $\workernum$ and $\budget$. Specifically, if $t> bk$, the running time of masters is linear in $\workernum$ and independent of $\budget$, while that of each worker is linear in $\budget$ and independent of $\workernum$. The running time of aggregators increases linearly with both $\workernum$ and $\budget$.
However, as given in Eq.~\eqref{var:approx:method}, increasing $\workernum$ reduces the variance faster than increasing $\budget$ does. Specifically, the third term in Eq.~\eqref{var:approx:method} decreases quadratically with $\workernum$, while it decreases linearly with $\budget$.

When using \trifly, however, $10$ workers with $10GB$ storage each is preferred. That is, we recommend increasing the storage budget per worker (i.e., $\budget$) first rather than the number of workers (i.e., $\workernum$).
This is because increasing $\budget$ reduces the variance faster than increasing $\workernum$ does. Specifically, the first term in Eq.~\eqref{var:approx:naive} decreases quadratically with $\budget$, while it decreases linearly with $\workernum$.
When $t$ is large enough, the elapsed time of \trifly increases linearly with both $\workernum$ and $\budget$. Specifically, as summarized in Table~\ref{tab:complexity}, the running time of masters is linear in $\workernum$ and independent of $\budget$, while that of each worker is linear in $\budget$ and independent of $\workernum$. The running time of aggregators increase linearly with both $\workernum$ and $\budget$.


The only remaining parameter is the tolerance threshold $\theta$ in \cocosopt. Based on the empirical results in Section~\ref{sec:trifly:experiments:param}, we recommend setting it to $0.2$.

\section{Experiments}
\label{sec:trifly:exp}
\begin{table}[t]
	\centering
	\caption[Summary of the graph streams used in our experiments]{\label{tab:trifly:data:real} \figsummary{Summary of the graph streams used in our experiments.} B: billion, M: million, K: thousand.
	}
		\begin{tabular}{l|r|r|l}
			\toprule
			{\bf Name} & {\bf \# Nodes} & {\bf \# Edges} & {\bf Summary} \\
			\midrule
			\arxivD \cite{gehrke2003overview} 
			& $34.5$K & $421$K & Citation network \\
			\facebookD \cite{viswanath2009evolution} 
			& $63.7$K & $817$K & Friendship network \\
			\googleD \cite{leskovec2009community} 
			& $875$K & $4.32$M & Web graph \\
			\berkstanD \cite{leskovec2009community} 
			& $685$K & $6.65$M & Web graph \\
			\youtubeD \cite{mislove2007measurement} 
			& $3.22$M & $9.38$M & Friendship network\\ 
			\flickrD \cite{mislove2007measurement}  
			& $2.30$M & $22.8$M & Friendship network\\
			\livejournalD \cite{mislove2007measurement}  
			& $4.00$M & $34.7$M & Friendship network\\
			\friendsterD \cite{yang2015defining}  
			& $65.6$M & $1.81$B & Friendship network\\
			\midrule
			\randomD (800GB) & $1$M & $0.1$B-$100$B & Synthetic graph \\
			\bottomrule
		\end{tabular}
\end{table}

We review our experiments for answering the following questions:
\bit
	\item {\bf Q1. Illustration of Theorems}: Does \cocos give unbiased estimates? How do their variances scale with the number of workers? 
	\item {\bf Q2. Speed and Accuracy}: Is \cocos faster and more accurate than baselines?
	\item {\bf Q3. Scalability}: Does \cocos scale linearly with the number of edges in the input stream?
	\item {\bf Q4. Effects of Parameters}: How do the number of workers, storage budget, and parameter $\theta$ affect the accuracy of \cocos?
\eit

\subsection{Experimental Settings}
\label{sec:trifly:experiments:settings}

\smallsection{Machines:} All experiments were conducted on a cluster of 40 machines with 3.47GHz Intel Xeon X5690 CPUs and 32GB RAM.

\smallsection{Datasets:} We used the graphs listed in Table~\ref{tab:trifly:data:real}. 
We ignored all self loops, parallel edges, and directions of edges.
We simulated graph streams by streaming the edges of the corresponding graph in a random order from the disk of the machine hosting the master. 

\smallsection{Implementations:} We implemented the following algorithms commonly in C++ and MPICH 3.1:
\bit
	\item {\bf \cocossimple} (Section~\ref{sec:trifly:method:algorithm}): proposed distributed streaming algorithms using the modulo function as the node mapping function $f$  (i.e., $f(x) = x \text{ mode } \workernum$). 
	\item {\bf \cocosopt} (Section~\ref{sec:trifly:method:opt}): proposed distributed streaming algorithms using Algorithm~\ref{alg:cocos:opt} as the node mapping function $f$.
	\item {\bf \trifly} (Section~\ref{sec:trifly:method:baseline}): baseline distributed streaming algorithm.
	\item {\bf \mascot} \cite{lim2018memory} and {\bf \triestimp} \cite{stefani2017triest}: state-of-the-art single-machine streaming algorithms.
\eit
Among potential competitors, we chose streaming algorithms that estimate both global and local triangle counts. The chosen algorithms, \mascot and \triestimp, are also more accurate than several well-known single-machine streaming algorithms that estimate only the global triangle count, as shown in Appendix~\ref{sec:trifly:appendix:single}.
For the distributed algorithms, we used one master and one aggregator hosted by the same machine.
Workers were hosted by different machines (unless their number was greater than that of machines). They used a part of the main memory of hosting machines as their local storage.
In every algorithm,
sampled edges were stored in the adjacency list format, and lazy aggregation, explained in Section~\ref{sec:trifly:method:lazy}, was used so that all estimates were aggregated once at the end of the input stream.
We fixed $\theta$ in \cocosopt to $0.2$, which gave the best accuracy (see Section~\ref{sec:trifly:experiments:param}).

\smallsection{Evaluation Metrics:}
We measured the accuracy of the considered algorithms at the end of each input stream.
Let $\SG=(\SV,\SE)$ be the graph at the end of the input stream.
Then, for each node $u\in\SV$, let $x[u]$ be the true local count of $u$ in $\SG$, and let $\hat{x}[u]$ be its estimate obtained by the evaluated algorithm.
Likewise, let $x$ and $\hat{x}$ be the true and estimated global triangle counts, respectively.\footnote{We computed the exact counts of global and local triangles in large-scale graphs, using \cocosopt with enough storage budget $b$.}
We evaluated 
the accuracy of global triangle counting using {\em global error}, defined as {\small$\frac{|x-\hat{x}|}{1+x}$}, and {\em global variance}, defined as {\small $(x-\hat{x})^2$}.\footnote{Note that all considered algorithms are unbiased. Note that, we can estimate the variance of $\hat{x}$ by computing these measure multiple times and then computing the mean of them.}
For the accuracy of local triangle counting,
we used {\em local error}, defined as {\small $\frac{1}{|\SV|}\sum\nolimits_{u\in\SV}\frac{|x[u]-\hat{x}[u]|}{1+x[u]}$}, and {\em local RMSE}, defined as {\small $\sqrt{\frac{1}{|\SV|}\sum\nolimits_{u\in\SV}(x[u]-\hat{x}[u])^{2}}$}. We also used Spearman's rank correlation coefficient \cite{spearman1904proof} between $\{(u,x[u])\}_{u\in\SV}$ and $\{(u,\hat{x}[u])\}_{u\in\SV}$.

\begin{figure}[t]
	\centering
	\includegraphics[width=0.72\linewidth]{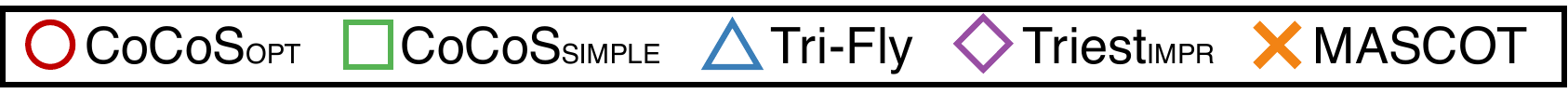}\\
	\vspace{-2.5mm}
	\subfigure[\arxivD Dataset]{
		\includegraphics[width= 0.24\linewidth]{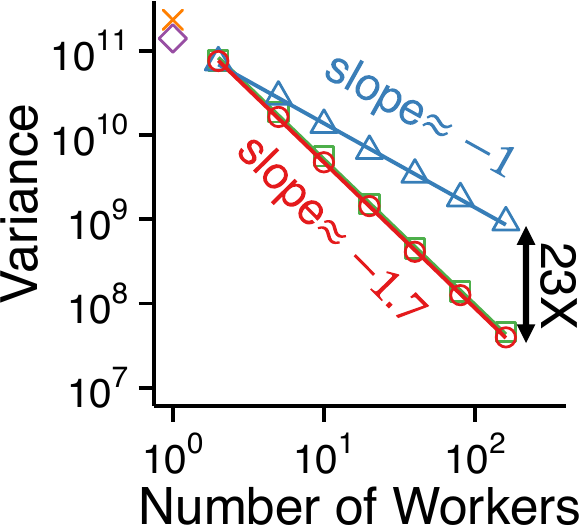}
	}
	\hspace{-1mm}
	\subfigure[\facebookD Dataset]{
		\includegraphics[width= 0.24\linewidth]{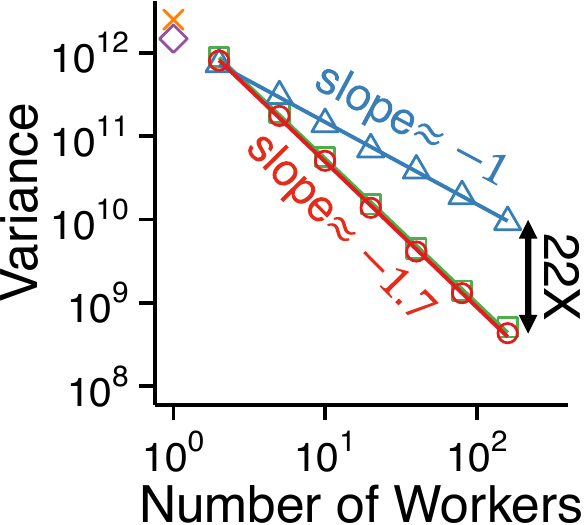}
	}
	\hspace{-1mm}
	\subfigure[\googleD Dataset]{
		\includegraphics[width= 0.24\linewidth]{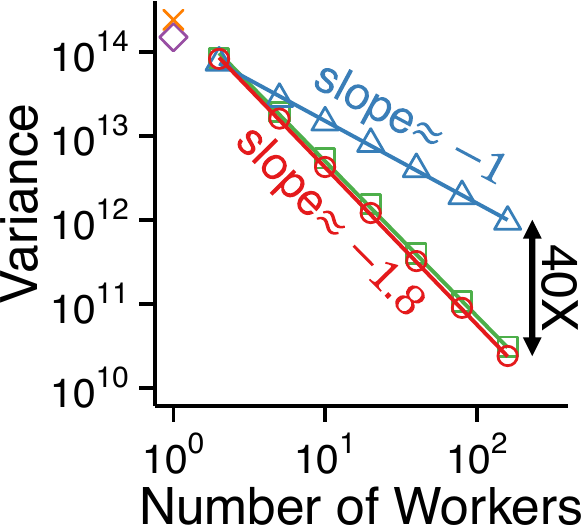}
	}
	\\ \vspace{-2mm}
	\caption[Variance Drop in \cocos]{\label{fig:trifly:theorem:variance}
		\figsummary{The variance of estimates drops faster in \cocosopt and \cocossimple than in \trifly}, as we use more workers.
	}
\end{figure}

\subsection{Q1. Illustration of Our Theorems}
\label{sec:trifly:experiments:theorem}

\smallsection{\cocos gave unbiased estimates with small variances.}
Figure~\ref{fig:trifly:crown:unbias} in Section~\ref{sec:trifly:intro} illustrates Theorems~\ref{thm:trifly:bias:naive} and \ref{thm:trifly:bias:method}, the unbiasedness of \trifly and \cocos.
We obtained $10,000$ estimates of the global triangle count in the \googleD dataset using each distributed algorithm.
We used $30$ workers, and set $\budget$ so that each worker stored up to $5\%$ of the edges.
As expected from Theorems~\ref{thm:trifly:bias:naive} and \ref{thm:trifly:bias:method},
\trifly, \cocosopt, and \cocossimple gave estimates whose averages were close to the true triangle count.
The variance was the smallest in \cocosopt, and the variance in \cocossimple was smaller than that in \trifly.

\smallsection{The variance in \cocos dropped fast with the number of workers.}
Figure~\ref{fig:trifly:theorem:variance} illustrates Theorems \ref{thm:trifly:var:naive} and \ref{thm:trifly:var:method}, the variances of the estimates of the global triangle count in \trifly and \cocos.
As we scaled up the number of workers, the variance decreased faster in \cocosopt and \cocossimple ($\approx\workernum^{-1.7}$) than in \trifly ($\approx\workernum^{-1}$), as expected in Eq.~\eqref{var:approx:method} and Eq.~\eqref{var:approx:naive} in Section~\ref{sec:trifly:analysis:accuracy:variance}.
In each setting, $\budget$ was set to $1,000$, and the variance was estimated from $1,000$ trials.

\begin{figure}[t]
	\centering
	\subfigure[\youtubeD Dataset]{
		\includegraphics[width= 0.24\linewidth]{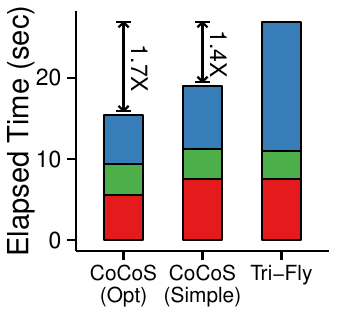}
	}
	\subfigure[\livejournalD Dataset]{
		\includegraphics[width= 0.24\linewidth]{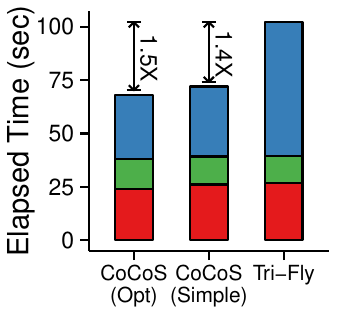}
	}
	\includegraphics[width= 0.22\linewidth]{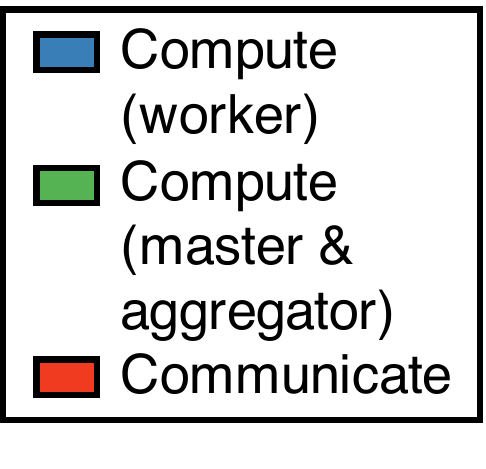}
	\\	\vspace{-2mm}
	\caption[Computation and communication overhead in \cocos]{\label{fig:trifly:performance:analysis} \figsummary{\cocosopt reduces both computation and communication overhead}, compared to \cocossimple and \trifly. \cocosopt is also more accurate than the others, as seen in Figure~\ref{fig:trifly:tradeoff}.}
\end{figure}

\begin{figure}
	\centering
	\includegraphics[width=0.9\linewidth]{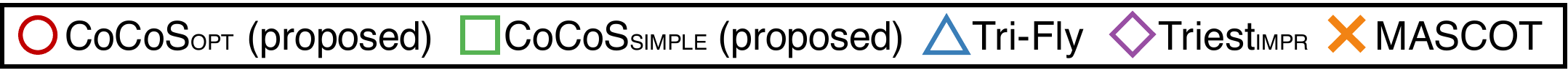}\\
	\vspace{0.5mm}
	\noindent\textbf{\globalerrorL}  (the lower the better): \hfill \ \ \ \\
	\hspace{-4mm}
	\vspace{0.5mm}
	\subfigure[\youtubeD]{
		\includegraphics[width= 0.185\linewidth]{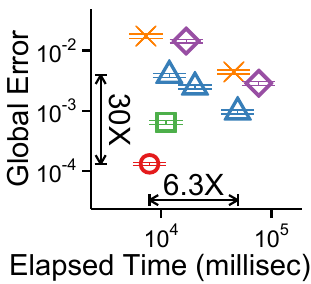}
	}
	\subfigure[\berkstanD]{
		\includegraphics[width= 0.185\linewidth]{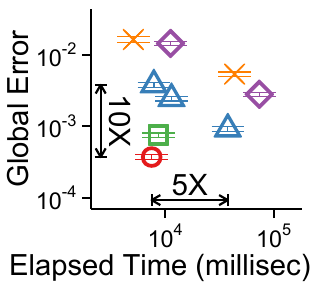}
	}
	\subfigure[\flickrD ]{
		\includegraphics[width= 0.185\linewidth]{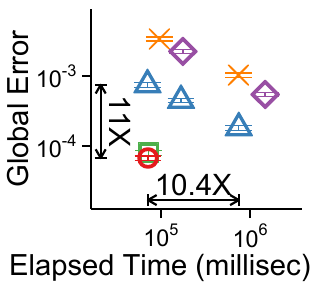}
	}
	\subfigure[\livejournalD]{
		\includegraphics[width= 0.185\linewidth]{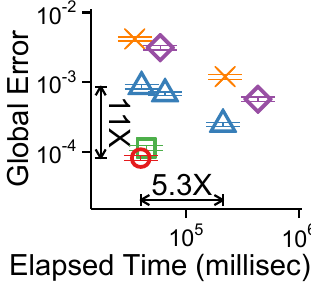}
	}
	\subfigure[\friendsterD]{
		\includegraphics[width= 0.185\linewidth]{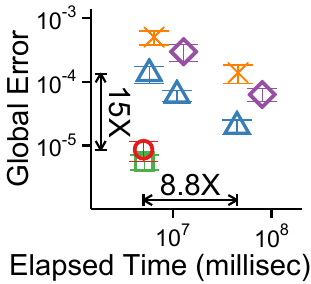}
	} \\
	\noindent\textbf{\globalvarL}  (the lower the better): \hfill \ \ \ \\
	\hspace{-4mm}
	\vspace{0.5mm}
	\subfigure[\youtubeD]{
		\includegraphics[width= 0.185\linewidth]{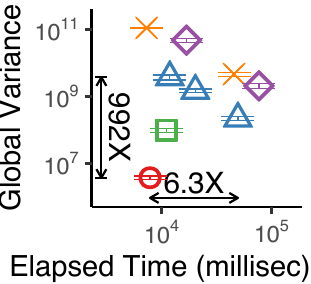}
	}
	\subfigure[\berkstanD]{
		\includegraphics[width= 0.185\linewidth]{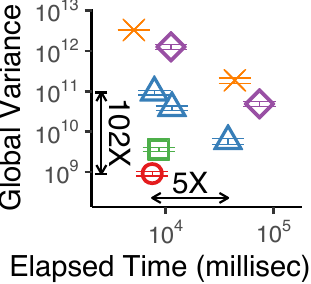}
	}
	\subfigure[\flickrD ]{
		\includegraphics[width= 0.185\linewidth]{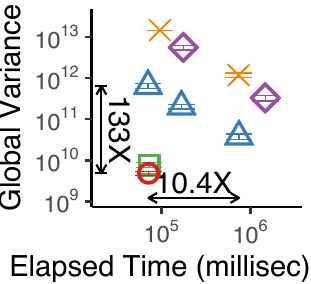}
	}
	\subfigure[\livejournalD]{
		\includegraphics[width= 0.185\linewidth]{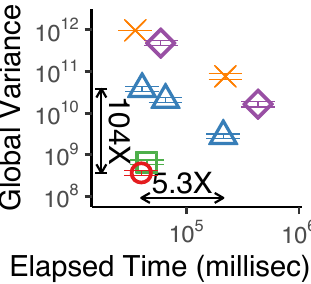}
	}
	\subfigure[\friendsterD]{
		\includegraphics[width= 0.185\linewidth]{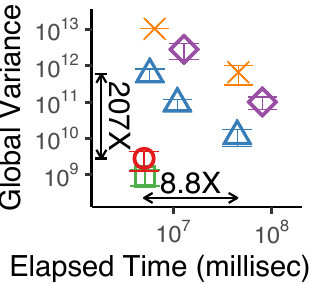}
	}
	\\
	\noindent\textbf{\localerrorL}  (the lower the better): \hfill \ \ \ \\
	\hspace{-4mm}
	\vspace{0.5mm}
	\subfigure[\youtubeD]{
		\includegraphics[width= 0.185\linewidth]{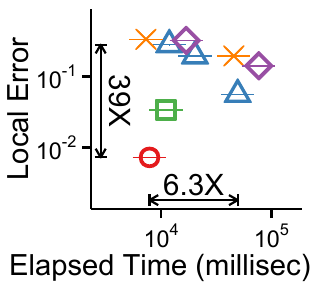}
	}
	\subfigure[\berkstanD]{
		\includegraphics[width= 0.185\linewidth]{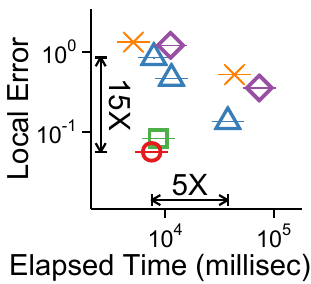}
	}
	\subfigure[\flickrD]{
		\includegraphics[width= 0.185\linewidth]{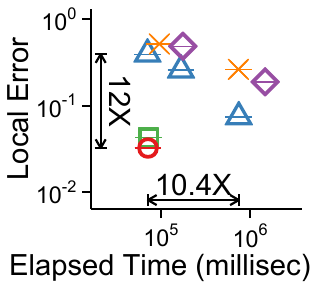}
	}
	\subfigure[\livejournalD]{
		\includegraphics[width= 0.185\linewidth]{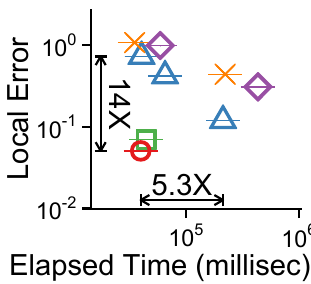}
	}
	\subfigure[\friendsterD]{
		\includegraphics[width= 0.185\linewidth]{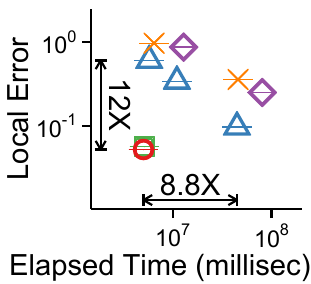}
	}
	\\
	\noindent\textbf{\rmse}  (the lower the better): \hfill \ \ \ \\
	\hspace{-4mm}
	\vspace{0.5mm}
	\subfigure[\youtubeD]{
		\includegraphics[width= 0.185\linewidth]{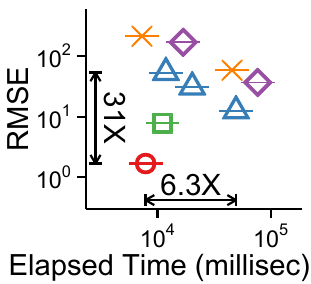}
	}
	\subfigure[\berkstanD]{
		\includegraphics[width= 0.185\linewidth]{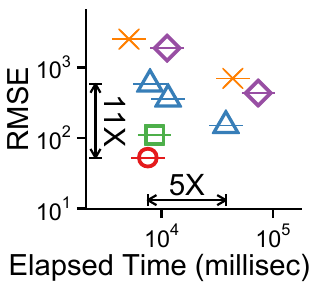}
	}
	\subfigure[\flickrD]{
		\includegraphics[width= 0.185\linewidth]{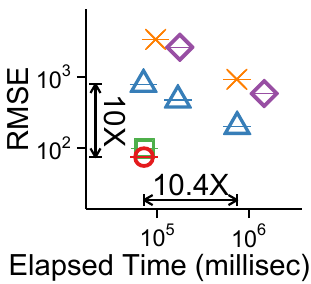}
	}
	\subfigure[\livejournalD]{
		\includegraphics[width= 0.185\linewidth]{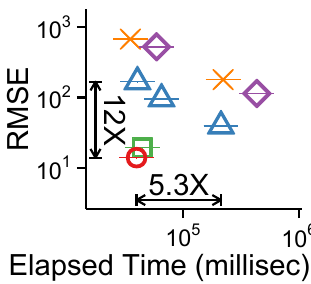}
	}
	\subfigure[\friendsterD]{
		\includegraphics[width= 0.185\linewidth]{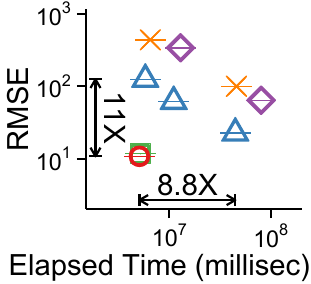}
	}
	\\
	\noindent\textbf{\rankcorrelationL}  (the higher the better): \hfill \ \ \ \\
	\hspace{-4mm}
	\vspace{0.5mm}
	\subfigure[\youtubeD]{
		\includegraphics[width= 0.185\linewidth]{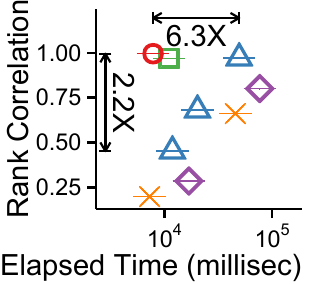}
	}
	\subfigure[\berkstanD]{
		\includegraphics[width= 0.185\linewidth]{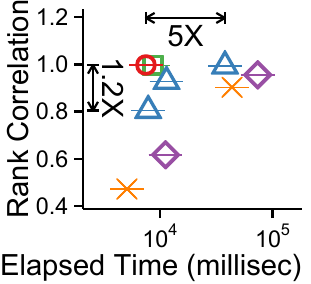}
	}
	\subfigure[\flickrD]{
		\includegraphics[width= 0.185\linewidth]{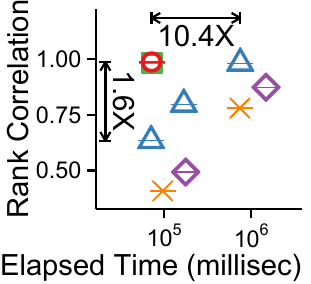}
	}
	\subfigure[\livejournalD]{
		\includegraphics[width= 0.185\linewidth]{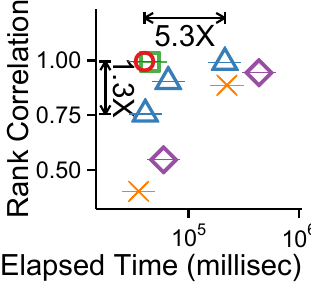}
	}
	\subfigure[\friendsterD]{
		\includegraphics[width= 0.185\linewidth]{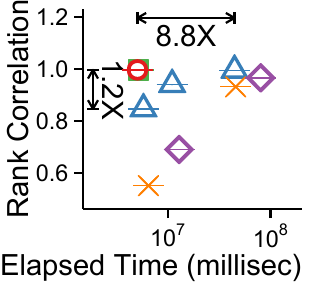}
	} \\
	\vspace{-2mm}
	\caption[Speed and accuracy of \cocos]{\label{fig:trifly:tradeoff} \figsummary{\cocos is fast and accurate.}
		\cocosopt (with $\theta$ fixed to $0.2$) yields estimates with up to $39\times$ smaller errors and $992\times$ smaller variances than those of the baselines with similar speeds, and it is up to $10.4\times$ faster than the baselines while offering higher accuracy.
		Error bars show sample standard errors.
	}
\end{figure}

\subsection{Q2. Speed and Accuracy}
\label{sec:trifly:experiments:tradeoff}

We measured the speed and accuracy of the considered algorithms with different storage budgets.\footnote{$\budget=5\%$ of the number of edges in each dataset in \cocossimple and \cocosopt. $\budget=\{2\%,$$5\%,20\%\}$ in \trifly. $\budget=\{5\%,40\%\}$ in \triestimp and \mascot.
See Section~\ref{sec:trifly:experiments:param} for the effects of $\budget$ values on the accuracies of the algorithms.
}
We used 30 workers for each distributed streaming algorithm.
To compare their speeds independently of the speed of the input stream, we measured the time taken by each algorithm to process edges, ignoring the time taken to wait for the arrival of edges in the input stream.
In Figure~\ref{fig:trifly:tradeoff}, we report the evaluation metrics and elapsed times averaged over $10$ trials in the \friendsterD dataset and over $100$ trials in the other large datasets.

\smallsection{\cocos gave the best trade-off between speed and accuracy.}
Specifically, \cocos was up to $\mathit{10.4\times}$ {\it faster} than the baselines while giving more accurate estimates.
Moreover, \cocos was up to $\mathit{30\times}$ and $\mathit{39\times}$ {\it more accurate} than the baselines with similar speeds in terms of global error and local error, respectively.
Moreover, \cocos yielded estimates of the global triangle count with up to $\mathit{992\times}$ smaller variances than those of the baselines with similar speeds.
Between the proposed algorithms, \cocosopt was up to $1.4\times$ faster and $4.9\times$ more accurate than \cocossimple. 

\smallsection{\cocosopt reduced computation and communication overhead.}
Figure~\ref{fig:trifly:performance:analysis} shows elapsed times for (a) computation in
the master and aggregator, (b) computation in the slowest worker, and (c) communication between machines in \cocosopt, \cocossimple, and \trifly.
The storage budget $\budget$ was set to $5\%$ of the number of edges in each dataset.
\cocosopt reduced computation and communication costs, compared to \cocossimple and \trifly, as we expect in Section~\ref{sec:trifly:method:opt}.
Recall that \cocosopt was also more accurate than \cocossimple and \trifly.

\subsection{Q3. Scalability}
\label{sec:trifly:experiments:scalable}

We measured how the running times of \cocosopt and \cocossimple scale with the number of edges in the input stream.
We used $30$ workers with $\budget$ fixed to $10^{7}$, and we measured their running times independently of the speed of the input stream, as in Section~\ref{sec:trifly:experiments:tradeoff}.

\smallsection{\cocos scaled linearly and handled terabyte-scale graphs.}
Figure~\ref{fig:trifly:scalability:random} shows the results in Erd\H{o}s-R\'enyi random graph streams with $1$ million nodes and different numbers of edges, 
and Figure~\ref{fig:trifly:scalability:real} shows the results in graph streams with realistic structures created by sampling different numbers of edges from the \friendsterD dataset.
Note that the largest stream has {\it $\mathit{100}$ billion edges}, which are {\it $\mathit{800}$GB.}
\cocosopt and \cocossimple scaled linearly with the size of the input stream, as we expect in Section~\ref{sec:trifly:analysis:complexity:time}.



\begin{figure}[t]
	\centering
	\subfigure[\randomD Dataset]{
		\label{fig:trifly:scalability:random} 
		\includegraphics[width= 0.24\linewidth]{FIG/crown_scalability.pdf}
	}
	\subfigure[\friendsterD Dataset]{
		\label{fig:trifly:scalability:real} 
		\includegraphics[width= 0.24\linewidth]{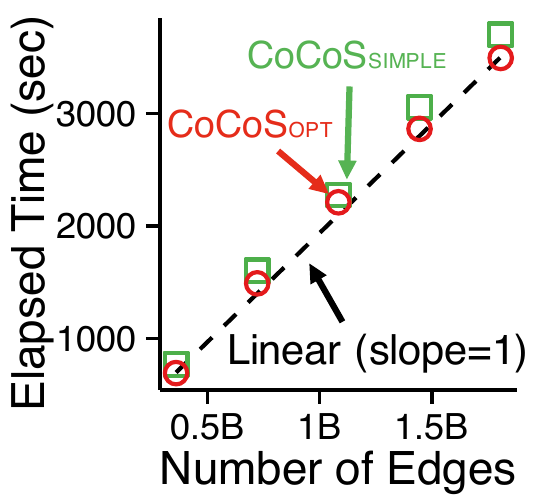}
	} \\
	\vspace{-2mm}
	\caption[Scalability of \cocos]{\label{fig:trifly:scalability} \figsummary{\cocosopt and \cocossimple scale to terabyate-scale streams linearly with the size of the input stream.}}
\end{figure}

\subsection{Q4. Effects of Parameters on Accuracy}
\label{sec:trifly:experiments:param}

We explored the effects of the parameters on the accuracies of the considered algorithms.
As a default setting, we used $30$ workers for the distributed streaming algorithms and set $\budget$ to $2\%$ of the number of edges for each dataset and $\theta$ to $0.2$. When the effect of a parameter was analyzed, the others were fixed to their default values. We reported results with global error as the evaluation metric but obtained consistent results with the other metrics.
We measured it $1,000$ times in each setting and reported the average.
In Figures~\ref{fig:trifly:worker}-\ref{fig:trifly:tolerance}, the error bars denote sample standard errors.

\begin{figure}[t]
	\centering
	\includegraphics[width=0.72\linewidth]{FIG/label_point_simple.pdf}\\
	\vspace{-2mm}
	\subfigure[\arxivD Dataset]{
		\includegraphics[width= 0.248\linewidth]{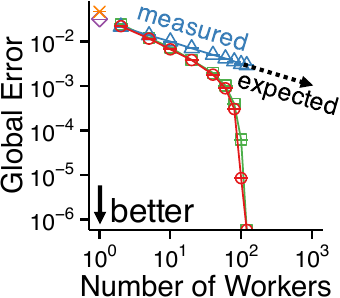}
	}
	\subfigure[\facebookD Dataset]{
		\includegraphics[width= 0.248\linewidth]{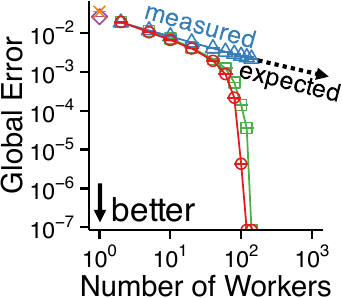}
	}
	\subfigure[\googleD Dataset]{
		\includegraphics[width= 0.24\linewidth]{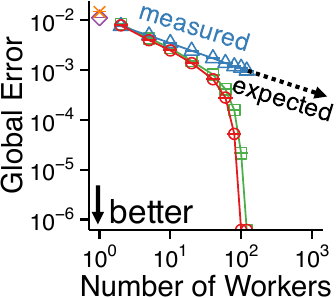}
	}	
	\\
	\vspace{-2mm}
	\caption[Effects of the number of workers on the accuracy of \cocos]{\label{fig:trifly:worker}
		\figsummary{Estimation error decreases faster in \cocos than in \trifly}, as we use more workers. 
	}
\end{figure}

\smallsection{As more workers were added, the estimation error decreased faster in \cocos} than in the baselines.
As seen in Figure~\ref{fig:trifly:worker}, the estimation errors of \cocosopt and \cocossimple became zero with about $100$ workers. However, that of \trifly dropped slowly with expectation that it never becomes zero with a finite number of workers (see Theorem~\ref{thm:trifly:var:naive}).

\begin{figure}[t]
	\centering
	\includegraphics[width=0.72\linewidth]{FIG/label_point_simple.pdf}\\
	\vspace{-2mm}
	\subfigure[\arxivD Dataset]{
		\includegraphics[width= 	0.24\linewidth]{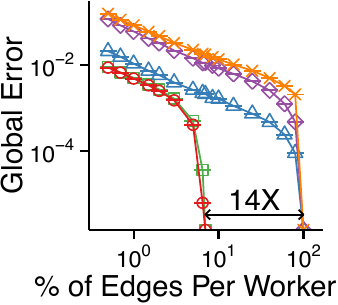}
	}
	\subfigure[\facebookD Dataset]{
		\includegraphics[width= 0.24\linewidth]{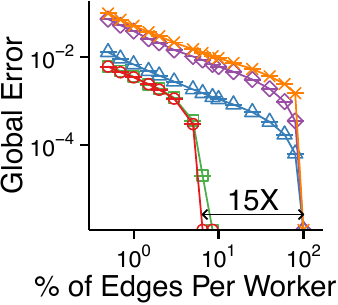}
	}
	\subfigure[\googleD Dataset]{
		\includegraphics[width= 0.24\linewidth]{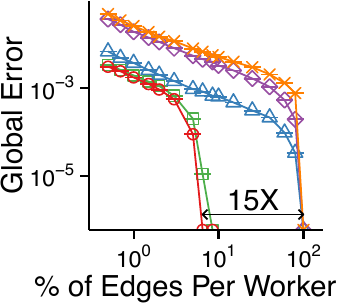}
	}
	\\
	\vspace{-2mm}
	\caption[Effects of the storage budget $b$ on the accuracy of \cocos]{\label{fig:trifly:sampling}
		\figsummary{Estimation error decreases faster in \cocos than in the baselines}, as we increase storage budget $\budget$. 
		For exact estimation,
		\cocos requires $14\times$ smaller $\budget$ than the others.
	}
\end{figure}

\smallsection{As storage budget increased, the estimation error decreased faster in \cocos} than in the baselines.
As seen in
Figure~\ref{fig:trifly:sampling}, the estimation errors of \cocosopt and \cocossimple became $0$ when each worker could store about $7\%$ of the edges in each dataset.
However, the estimation errors of the baselines became zero only when each worker could store all the edges in each dataset.

\begin{figure}[t]
	\centering
	\subfigure[\arxivD Dataset]{
		\includegraphics[width= 0.24\linewidth]{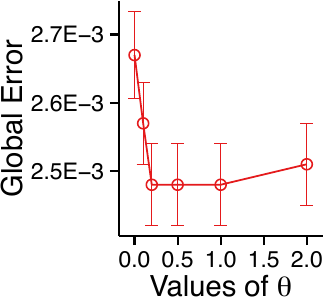}
	}
	\subfigure[\facebookD Dataset]{
		\includegraphics[width= 0.24\linewidth]{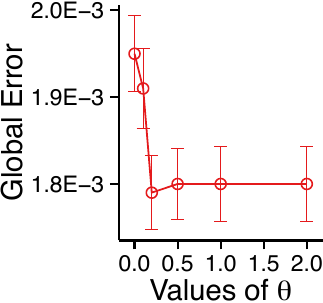}
	} 
	\subfigure[\googleD Dataset]{
		\includegraphics[width= 0.24\linewidth]{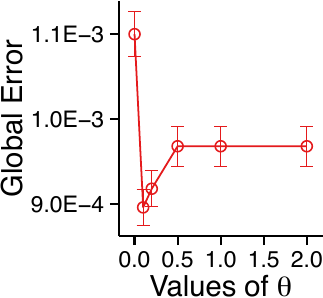}
	} \\
	\vspace{-2mm}
	\caption[Effects of the tolerance $\theta$ on the accuracy of \cocos]{\label{fig:trifly:tolerance}
		\figsummary{Estimation error in \cocosopt is smallest when $\theta$
			is around $0.2$,} while the error is not very sensitive to $\theta$.
	}
\end{figure}

\smallsection{\cocosopt was most accurate when $\theta$ was around $\mathbf{0.2}$}, as
seen in Figure~\ref{fig:trifly:tolerance}. The estimation error, however, was not very sensitive to the value of $\theta$ as long as $\theta$ was at least $0.2$.

\section{Conclusions}
\label{sec:trifly:summary}
In this work, we propose \cocos, a fast and accurate distributed streaming algorithm for the counts of global and local triangles.
By minimizing the redundant use of distributed computational and storage resources (P1-P3 in Lemma~\ref{lemma:trifly:property}),
\cocos offers the following advantages:
\bit
	\item {\bf Accurate}: \cocos is up to {\it $\mathit{39\times}$ more accurate} than its similarly fast competitors (Figure~\ref{fig:trifly:tradeoff}). It gives exact estimates within {\it $\mathit{14\times}$ smaller storage budgets} than its competitors (Figure \ref{fig:trifly:sampling}).
	\item {\bf Fast}: \cocos is up to {\it $\mathit{10.4\times}$ faster} than its competitors while giving more accurate estimates (Figure~\ref{fig:trifly:tradeoff}). 
	\cocos scales linearly with the size of the input stream (Figure~\ref{fig:trifly:scalability}). 
	\item {\bf Theoretically Sound}: \cocos gives unbiased estimates (Theorem~\ref{thm:trifly:bias:method}).
\eit
{\bf Reproducibility:} The source code and datasets used in this chapter are available at \textit{\url{http://dmlab.kaist.ac.kr/cocos/}}.

\section*{Acknowledgments}
This research was supported by Disaster-Safety Platform Technology Development Program of the National Research Foundation of Korea (NRF) funded by the Ministry of Science and ICT (Grant Number: 2019M3D7A1094364)
and Institute of Information \& Communications Technology Planning \&
Evaluation (IITP) grant funded by the Korea government (MSIT) (No. 2019-0-00075, Artificial Intelligence Graduate School Program (KAIST)).
This research was also supported by the National Science Foundation under Grant No. CNS-1314632 and IIS-1408924.
Research was sponsored by the Army Research Laboratory 
and was accomplished under Cooperative Agreement Number W911NF-09-2-0053. 
This publication was made possible by NPRP grant \# 7-1330-2-483 from the Qatar National Research Fund (a member of Qatar Foundation).
Any opinions, findings, and conclusions or recommendations expressed in this
material are those of the author(s) and do not necessarily reflect the views
of the National Science Foundation, or other funding parties.
The U.S. Government is authorized to reproduce and 
distribute reprints for Government purposes notwithstanding 
any copyright notation here on.

\appendix
\section{Appendix: Proof of Lemma~3}
\label{sec:trifly:appendix:proof}

\begin{proof}
For each triangle $\triple\in \STT$ with $\tvw < \twu < \tuv \leq t$, 
let $\fuvw\in\workerset$ be the worker that can possibly count $\triple$.
That is, $\fuvw = f(w)$ if $f(u) \neq f(v)$, and $\fuvw = f(u) = f(v)$ otherwise. 

For the first claim, note that for each triangle $\triple$, each worker has the equal probability of being $\fuvw$. 
Therefore, 
\begin{equation}
\BE[\globalnumi] = \frac{\globalnum}{\workernum}. \label{eq:exp:triangle:proof} 
\end{equation}

For the second claim, for each edge $\uv$, the probability that it is assigned to each $i$-th worker is equal to the probability that $f(u) = i$ or $f(v) = i$, which is 
$1-\left(1-\frac{1}{\workernum}\right)^{2}=\frac{2\workernum - 1}{\workernum^2}$.
Therefore,
\begin{equation}
\BE[\lti] =  \frac{(2\workernum-1)t}{\workernum^2} \label{eq:exp:load:proof}
\end{equation}

\begin{figure}[h]
	\centering
	\vspace{-3mm}
	\subfigure[Type~1 Pair]{
		\includegraphics[width=0.47\linewidth]{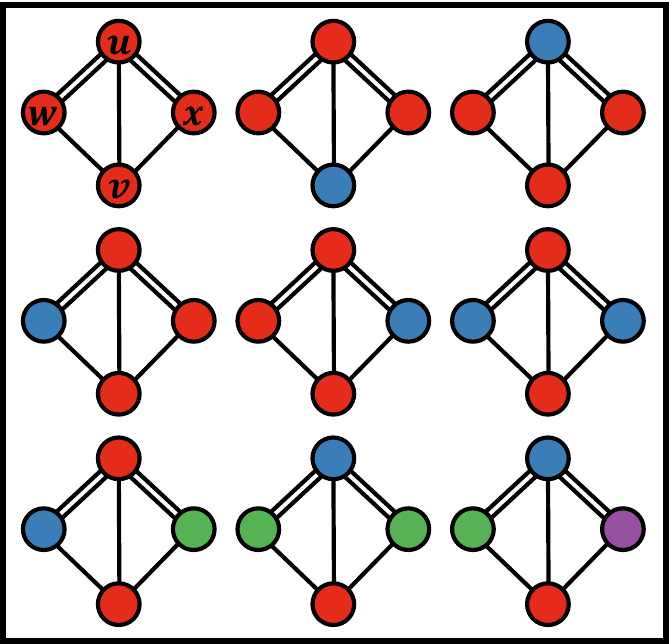}
		\label{fig:trifly:coloring:type1} 
	}
	\subfigure[Type~2 Pair]{
		\includegraphics[width=0.47\linewidth]{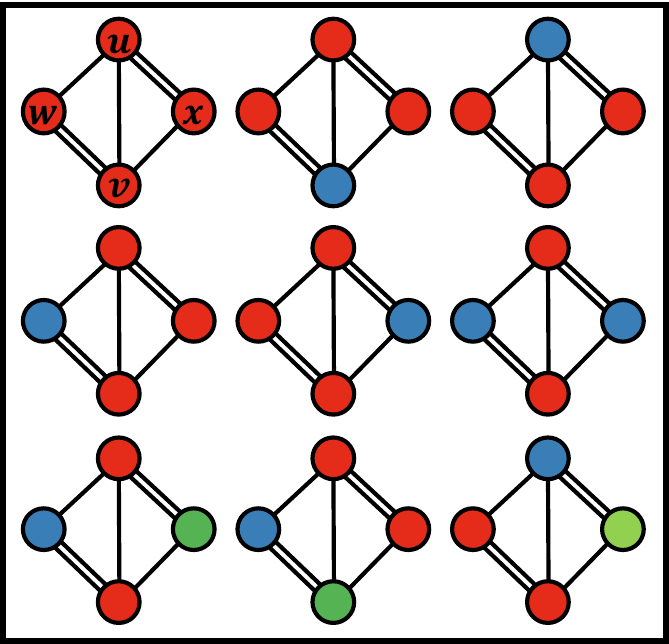}
		\label{fig:trifly:coloring:type2} 
	} \\
	\vspace{-2mm}
	\caption[Coloring of Type~1 and Type~2 triangle pairs where $f(uvw)=f(uvx)$]{\figsummary{Coloring of (a) Type~1 and (b) Type~2 triangle pairs where $\mathbf{f(uvw)=f(uvx)}$.}
		Nodes assigned to worker $f(uvw)$ ($= f(uvx)$) by $f$ are colored red.
		Nodes with different colors are assigned to different workers by $f$.
	}
\end{figure}

For the third claim, consider a Type 1 triangle pair $\triple$ and $\tripletwo$. 
By considering $f : \SV \to \workerset$ as a coloring of nodes $\SV$ with $\workernum$ colors, 
Figure~\ref{fig:trifly:coloring:type1} represents all the nine ways where $f(uvw) = f(uvx)$.
Note that $f(uvw) = f(uvx)$ is colored red in all of them. 
Fix a worker $i \in \workerset$. Then,  
\begin{align*}
	P[f(uvw) = f(uvx)=i] = &  \frac{1}{\workernum^4} + \frac{6}{\workernum^3} (1 - \frac{1}{\workernum})  + \frac{2}{\workernum^2} (1  - \frac{1}{\workernum})(1  - \frac{2}{\workernum}) \\
	&  + \frac{1}{\workernum}(1 - \frac{1}{\workernum}) (1 - \frac{2}{\workernum}) (1 - \frac{3}{\workernum}),
\end{align*}
where each term from left to right in the right hand side corresponds to the 1st case, 2nd-6th cases, 7th-8th cases, and 9th case, respectively, in Figure~\ref{fig:trifly:coloring:type1}. Therefore,
\begin{equation}
	\BE[\pti] = P[f(uvw) = f(uvx)=i]\pt = \frac{ \workernum^3 - 4 \workernum^2 + 10 \workernum - 6}{ \workernum^4 } \pt. \label{eq:exp:p:proof} 
\end{equation}

For the fourth claim, consider a Type 2 triangle pair $\triple$ and $\tripletwo$. 
By considering $f : \SV \to \workerset$ as a coloring of nodes $\SV$ with $\workernum$ colors, 
Figure~\ref{fig:trifly:coloring:type2} represents all the nine ways where $f(uvw) = f(uvx)$.
Note that $f(uvw) = f(uvx)$ is colored red in all of them. 
Fix a worker $i \in \workerset$. Then,  
\begin{equation*}
	P[f(uvw) = f(uvx)=i] = \frac{1}{\workernum^4} + \frac{1}{\workernum^3}(1 - \frac{1}{\workernum}) + \frac{1}{\workernum^2}(1  - \frac{1}{\workernum})(1  - \frac{2}{\workernum}),
\end{equation*}
where each term from left to right in the right hand side corresponds to the 1st case, 2nd-6th cases, and 7th-9th cases, respectively, in Figure~\ref{fig:trifly:coloring:type2}.
Therefore, 
\begin{equation}
	\BE[\qti] = P[f(uvw) = f(uvx)=i]\qt =\frac{3 \workernum^2 - 4 \workernum + 2}{ \workernum^4 } \qt. \label{eq:exp:q:proof} 
\end{equation}
Eq.~\eqref{eq:exp:triangle:proof}, Eq.~\eqref{eq:exp:load:proof}, Eq.~\eqref{eq:exp:p:proof}, Eq.~\eqref{eq:exp:q:proof}, $k\geq  1$, and $\epsilon=1$ imply Eq.~\eqref{eq:trifly:expectation:first} and Eq.~\eqref{eq:trifly:expectation:second}.
\end{proof}

\section{Appendix: A Comparison of Single-machine Streaming Algorithms}
\label{sec:trifly:appendix:single}

In Figure~\ref{fig:single}, we compare the accuracies of \triestimp \cite{stefani2017triest}, \mascot \cite{lim2018memory}, (parallel) Neighborhood Sampling (\ns) \cite{pavan2013counting,pavan2013parallel,tangwongsan2013parallel}, and Graph Sample and Hold (\gsh) \cite{ahmed2014graph}, while setting their storage budget so that up to $5\%$ of the edges in each dataset is stored.
\triestimp was most accurate among the single-machine streaming algorithms.
This result justifies our choice of adapting \triestimp for triangle counting in each worker.
Moreover, we lose good properties of \cocos and \trifly if they are equipped with the other algorithms rather than \triestimp. For example, combining them with \ns or \gsh does not support local triangle counting, and combining them with \mascot or \gsh requires prior knowledge about the input stream to set their parameters properly.

\begin{figure}[h]
	\centering
	\includegraphics[width=0.54\linewidth]{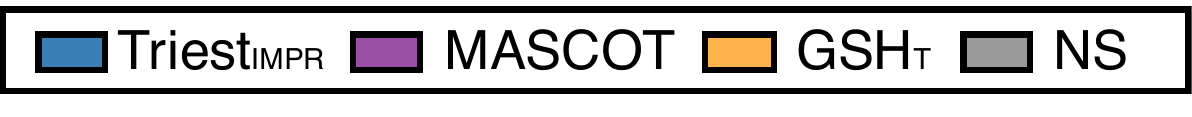} \\
	\includegraphics[width=0.77\linewidth]{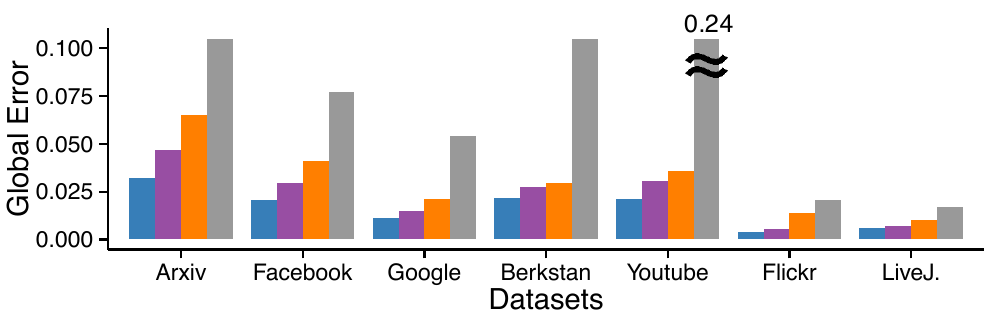} \\
	\vspace{-1mm}
	\caption{\label{fig:single} Among the considered single-machine streaming algorithms, \triestimp is most accurate.}
\end{figure}

\bibliographystyle{ACM-Reference-Format}
\bibliography{BIB/kijung}

\end{document}